\documentclass[11pt]{imsart}

\RequirePackage{amsthm,amsmath,amsfonts,amssymb}
\RequirePackage[numbers,sort&compress]{natbib}
\RequirePackage[colorlinks,citecolor=blue,urlcolor=blue]{hyperref}
\RequirePackage{graphicx}

\theoremstyle{plain}

\newtheorem{theorem}{Theorem}[section]
\newtheorem{lemma}[theorem]{Lemma}
\newtheorem{corollary}{Corollary}
\newtheorem{proposition}{Proposition}

\theoremstyle{remark}

\usepackage{xcolor}
\usepackage{color}

\definecolor{darkgreen}{rgb}{0, 0.4,0}

\definecolor{purplebrown}{rgb}{0.5,0.1,0.6}

\definecolor{ultclupcol}{rgb}{0.1,0.5,0.5}

\definecolor{mytrycolor}{rgb}{0.5,0.7,0.2}

\definecolor{ultclupcola}{rgb}{.5,0,.5}

\newcommand{\bl}[1]{\textcolor{blue}{#1}}

\newcommand{\prp}[1]{\textcolor{purple}{#1}}
\newcommand{\yellow}[1]{\textcolor{yellow}{#1}}
\newcommand{\green}[1]{\textcolor{green}{#1}}

\definecolor{shadebrown}{rgb}{0.1,0.1,0.9}
\definecolor{lightblue}{rgb}{0.2,0,1}

\usepackage{fancybox}
\usepackage{graphicx}
\usepackage{epstopdf}
\usepackage{wrapfig}

\usepackage{xcolor}
\usepackage{tcolorbox}
\tcbuselibrary{skins}

\definecolor{Red}{rgb}{0.00, 0.00, 0.00}

\newtcbox{\xmybox}{on line,
arc=7pt,
before upper={\rule[-3pt]{0pt}{10pt}},boxrule=0pt,
boxsep=0pt,left=6pt,right=6pt,top=0pt,bottom=0pt,enhanced, coltext=blue, colback=white!10!yellow}

\newtcbox{\xmyboxa}{on line,
arc=7pt,
before upper={\rule[-3pt]{0pt}{10pt}},boxrule=0pt,
boxsep=0pt,left=6pt,right=6pt,top=0pt,bottom=0pt,enhanced, colback=white!10!yellow}

\newtcbox{\xmyboxb}{on line,
arc=7pt,
before upper={\rule[-3pt]{0pt}{10pt}},boxrule=1pt,colframe=darkgreen!100!blue,
boxsep=0pt,left=6pt,right=6pt,top=0pt,bottom=0pt,enhanced, colback=white!10!yellow}

\newtcbox{\xmyboxc}{on line,
arc=7pt,
before upper={\rule[-3pt]{0pt}{10pt}},boxrule=.7pt,colframe=blue!100!blue,
boxsep=0pt,left=6pt,right=6pt,top=0pt,bottom=0pt,enhanced, coltext=blue, colback=white!10!yellow}

\newtcbox{\xmytboxa}{on line,
arc=7pt,
before upper={\rule[-3pt]{0pt}{10pt}},boxrule=.0pt,colframe=pink!50!yellow,
boxsep=0pt,left=6pt,right=6pt,top=0pt,bottom=0pt,enhanced, coltext=white, colback=blue!40!red}

\newtcbox{\xmytboxb}{on line,
arc=7pt,
before upper={\rule[-3pt]{0pt}{10pt}},boxrule=.0pt,colframe=pink!50!yellow,
boxsep=0pt,left=6pt,right=6pt,top=0pt,bottom=0pt,enhanced, coltext=white, colback=white!40!green}

\def\y{{\bf y}}

\def\u{{\bf u}}

\def\y{{\bf y}}

\def\tr{\mbox{Tr}}

\def\tr{{\rm tr}\,}

\def\diag{{\rm diag}\,}

\def\be{\begin{equation}}
\def\ee{\end{equation}}
\def\ba{\left[\begin{array}}
\def\ea{\end{array}\right]}

\def\u{{\bf u}}

\def\y{{\bf y}}

\def\1{{\bf 1}}

\def\0{{\bf 0}}

\def\vecw{\mbox{vec}}
\def\rankw{\mbox{rank}}
\def\diag{\mbox{diag}}

\def\bU{\bar{U}}

\def\bV{\bar{V}}

\def\mR{{\mathbb R}}
\def\mC{{\mathbb C}}

\def\mE{{\mathbb E}}
\def\mP{{\mathbb P}}

\def\calV{{\cal V}}
\def\calU{{\cal U}}

\def\lp{\left (}
\def\rp{\right )}

\begin{document}

\begin{frontmatter}
\title{
{\it Exact} Error in Matrix Completion: Approximately Low-Rank Structures and Missing Blocks}





\author{
	Agostino Capponi \thanks{Email: ac3827@columbia.edu, Department of Industrial Engineering and Operations Research, Columbia University, 500 West 120th Street, New York, NY 10027.}
	\and
	Mijailo Stojnic \thanks{
		Email: flatoyer@gmail.com, Department of Industrial Engineering and Operations Research, Columbia University, 500 West 120th Street, New York, NY 10027.}
}

\maketitle

\begin{abstract}
	We study the completion of approximately low rank matrices with entries \emph{missing not at random} (MNAR). In the context of typical large-dimensional statistical settings, we establish a framework for the performance analysis of the nuclear norm minimization ($\ell_1^*$) algorithm. Our framework produces  \emph{exact} estimates of the \emph{worst-case} residual root mean squared error (\textbf{RMSE}) and the associated \emph{phase transitions} (PT), with both exhibiting remarkably simple characterizations.  Our results enable to {\it precisely} quantify the impact of key system parameters, including data heterogeneity, size of the missing block, and deviation from ideal low rankness, on the accuracy of $\ell_1^*$-based matrix completion. To validate our theoretical worst-case RMSE estimates, we conduct numerical simulations, demonstrating close agreement with their numerical counterparts.
	
	
\end{abstract}



\end{frontmatter}





\section{Introduction}
\label{sec:intro}

Since the inception of matrix completion (MC), a large body of research (see, e.g., \cite{CR09matcomp,Rechtmatcomp11,CPmatcomp10,SAT05}) has focused on the completion when entries are missing at randomly (MAR) -- the most popular example being the completion of the movie-ratings matrix from the Netflix problem. In this paper, we study a variant of this problem which  has received significant attention in recent years (see, e.g., \cite{Agarwal2021,ABDIK21,XPlatfac10,BaiNg19,CahBaiNg21,DuPeXi23,KallusNIPS,AniShSh2020,MuhShSh2018}), namely the low rank completion (recovery) with entries {\it missing not at random} (MNAR).  Specifically, we consider a missingness pattern which follows a \emph{block structure}, common in many social and health science applications. For instance, suppose we want to estimate whether liver transplant surgery increases the life expectancy of individuals with cirrhosis. Then the liver transplant is an irreversible treatment, and estimating (\emph{causally inferring}) the counterfactual entries of the treated group (i.e., individuals subject to liver transplant surgery) would amount to estimating the block of entries where rows correspond to treated patients and columns to times since the treatment started. Another important example of block recovery arises in financial time series analysis, where many accounting based variables are missing after having been previously observed, leading to missing time blocks (see \cite{Pelgermissingness}).

The non-randomness of the missing entries invalidates the use of standard MC analytical methods. Exploring the block structure on the other hand, enables us to develop a novel analytical method and use it to characterize the performance of the nuclear norm minimization heuristic (which we denote by $\ell_1^*$) when employed for matrix completion and data imputation. However, some major challenges are faced along the way. A bit of contextualization helps to better understand them.

 \underline{\emph{Main challenges:}}  The performance analyses of the MC algorithms are usually done on a \emph{qualitative/scaling} level. Such analyses provide useful descriptive characterizations that unfortunately lack a full precision. The alternative, \emph{precise} ones, are much harder to obtain and typically  fall under the umbrella of the so-called \emph{phase-transition} (PT) considerations. Over the past couple of decades,  considerable  effort has been dedicated to obtaining the PTs of MC and their more general counterparts in low-rank recovery (LRR). This effort has yielded substantial success, as evidenced by studies such as \cite{DonGavMon13, Cinfidealwc22, OH10, CandesRecht11simple}. However, (with the exception of \cite{CandesRecht11simple} which, on the other hand, is a bit more limited in the range of covered parametric space) all these methods typically relate to, analytically slightly easier, LRR problems and usually rely on an underlying randomness in acquiring (or missing to acquire) the data. In other words, they relate to the LRR with the above mentioned MAR type of data acquisition paradigm.
 Despite the significant progress made in these directions, four major challenges have persisted as formidable hurdles: 
 
 (\textbf{\emph{i}}) Providing precise PT performance characterizations, visibly lacking in existing research, for \emph{truly} MNAR  in the ideal low rank MC; 
 (\textbf{\emph{ii}}) Bridging the gap between MNAR low-rank recovery (LRR) results and their counterparts in MC; (\textbf{\emph{iii}}) 
Establishing a connection between ideal and approximately low rank  MNAR MC;
 and 
 (\textbf{\emph{iv}}) Fully characterizing the MNAR MC scenarios with matrices whose internal structure slightly deviates from the ideal low rankness.

 \underline{\emph{Main contributions:}} We develop a novel  mathematical machinery that helps address all of the four above mentioned challenges. To precisely state our contributions, a few technical preliminaries are needed. In particular, we assume the existence of a partially observed (with unobserved entries forming a block) \emph{approximately low rank} ground-truth matrix $X_{sol}$. This implies that $X_{sol}$ is characterized by a limited number of dominant singular values, which primarily define/describe its heterogeneity, and a larger number of smaller singular values. Under mild technical assumptions stated in Section~\ref{sec:cinfreleqv}, we consider the use of the $\ell_1^*$ heuristic for completion/recovery of $X_{sol}$.  Assuming that $\hat{X}$ is the $X_{sol}$'s estimate produced by $\ell_1^*$, our main focus is studying the behavior of $\textbf{RMSE}_{(ns)}\triangleq\|\hat{X}-X_{sol}\|_F$. The main contributions of our study are as follows:  \textbf{\emph{(i)}} We consider a statistical approach with minimal distributional assumptions and leverage results from random matrix theory (in particular, free probability theory (FPT) \cite{Voic86,Voic87,Voic91}), to solve the problem of characterizing the performance of $\ell_1^*$ in the large dimensional approximately low rank MC scenario. \textbf{\emph{(ii)}} We show that $\ell_1^*$ exhibits a sharp  \emph{phase-transitioning} (PT) behavior, characterized by a boundary that separates regions of system parameters where the finiteness of the $\textbf{RMSE}$ is guaranteed or not. Consequently, we present a precise phase transition type of analysis (where the qualitative/scaling types of estimates are not allowed) and determine the \emph{exact} location of the regions separating PT curve.  
 \textbf{\emph{(iii)}} We show that the allowable \emph{data heterogeneity}, which ensures the finiteness of the $\textbf{RMSE}$, grows linearly with matrix size. Furthermore, our analysis precisely determines the involved proportionality constants. \textbf{\emph{(iv)}} We obtain  \emph{explicit exact worst case} estimates for the (appropriately scaled) residual $\textbf{RMSE}$ and precisely explain the connections/relations between the obtained $\textbf{RMSE}$s and PTs.

  We conduct a set of numerical experiments which demonstrate a strong agreement between the theoretical and simulated values, even for rather small matrix dimensions  (of order of a few tens) .

\subsection{Related prior work}
\label{sec:priorwwork}

We position our work with respect to three streams of  MC literature: (\textbf{\emph{i}}) general nuclear norm MC; (\textbf{\emph{ii}}) MC with observations \emph{truly} missing not at random (MNAR); and (\textbf{\emph{iii}}) performance analysis of nuclear norm optimization in \emph{true} MNAR scenarios.

(\textbf{\emph{i}}) The pioneering studies \cite{RFPrank,CR09matcomp,SAT05} introduced the nuclear norm, in the ideal low rank context, as a computationally efficient alternative to the highly nonconvex lowest rank minimization, and provided its generic performance analysis. Follow-up studies on algorithmic and performance characterization include, e.g., \cite{CT10matcomp,KLT11matcomp,KMO10matcomp,RT11matcomp,MHT10,NW11matcomp}. Many generic MC studies, however, rely on various forms of evenly balanced missingness most efficiently characterized through the so-called missing completely at random (MCAR) observations. On the other hand, many practically relevant observational scenarios are not of MCAR type. One, then, typically distinguishes two settings: 1) observations are still missing at random (MAR), but not completely at random and potentially dependently on other observations; 
 and 2) the missing patterns are either fully deterministic (we say, slightly differently from \cite{RubinInfMiss76}, of \emph{true} MNAR type) or strongly resemble a fully deterministic structure. The MAR type of observations have been extensively studied recently with strong algorithmic and theoretical results regarding consistency and residual errors. See \cite{FoySre11,Klopp14matcomp,CaiZhou16,SrShr05}, where the effects of the nonuniformly distributed missingness are studied on standard algorithms,  and \cite{ZhWaSa19,WaWoMaChCh21,BhCh21,SpBoJo18,SchSwSiChJo16,MaChen19}, where various algorithmic adaptations along the lines of IPW, primePCA, modified USVT and Nuclear norm, EM/FISTA, inverse propensity scoring (IPS), are respectively introduced and their performances analyzed to account for particular missingness types.

(\textbf{\emph{ii}}) The \emph{true} MNAR scenarios are analytically very hard and the available results are scarce. To the best of our knowledge (with the exception of our own \cite{Cinfidealwc22}\footnote{Some of the results appearing here, particularly those related to the ideal low rank scenario, were initially presented in a preliminary form in the conference proceedings  \cite{Cinfidealwc22}.}), we are the first to analyze a true MNAR instance on a \emph{phase transition} level of precision (where scaling types of estimates, $O(\cdot)$, are not allowed). 
Nonetheless, it is worth mentioning several works that study related problems and obtain different types of results. For example, \cite{Agarwal2021} considers various MNAR scenarios, introduces specific algorithmic designs, and theoretically analyzes the consistency of the individual matrix entries estimators (similar results were obtained in \cite{KallusNIPS,AniShSh2020,MuhShSh2018}). Recently, methods akin to MC have been employed for factor modeling in finance. In particular, \cite{XPlatfac10} discusses the role of PCA and the resulting consistency estimates in several MNAR settings, including both the block and the staggered ones (see also \cite{DuPeXi23} for further extensions in these directions). Related lines of work \cite{BaiNg19} and \cite{CahBaiNg21} introduce different types of algorithms, called Tall-wide and Tall-project, and prove their consistency estimates.

(\textbf{\emph{iii}}) The most closely related studies to our work are \cite{ABDIK21} and \cite{Cinfidealwc22}. In \cite{ABDIK21}, a noisy emulation of an approximately low rank scenario is used to obtain {qualitative/scaling} estimates of the nuclear norm's $\textbf{RMSE}$ upper bounds. It is not an MNAR setting, because it uses randomness in treatment selection but allows for a staggered scenario analysis (for the relevance of such a scenario, see, e.g., \cite{AthSte02,ShaTou19,AthImb18}). However, the study also discusses the use of the nuclear norm for block missing entries. 
In \cite{Cinfidealwc22}, we consider an idealized setup where ground truth matrices have low rank (but \textbf{\emph{not}} approximately low rank) and provide phase transition results.

\underline{\emph{Key Highlights:}} We highlight two major departures of our work relative to existing studies in the area: \textbf{\emph{(i)}} (\emph{``noisy vs approximately low rank"}) The majority of  surveyed studies uses low rank matrix plus noise to emulate approximately low rank settings. While in a qualitative/scaling type of analysis such an emulation can be appropriate, in phase transition types of analysis such an emulation may lead to incorrect conclusions. On the other hand, our model is designed to reflect real-word datasets, such as the  California tobacco data widely used in many studies (e.g.\cite{ADHsynth10}). In such data sets, deviations from ideal low rankness are directly attributable to the data's inherent structure, rather than to  inaccuracies stemming from erroneous recording. \textbf{\emph{(ii)}} (\emph{``homogeneous vs heterogeneous"}) Most, if not all, of the above mentioned works either assume fixed rank ($k$) of the dominating part of the matrix or allow for a weak dependence with the matrix size ($n$), i.e. allow for $k$ such that $\beta=\lim_{n\rightarrow\infty} \frac{k}{n}\rightarrow 0$. This means that their setups can only capture extremely homogeneous data settings. In many application contexts, including finance and health sciences, the large amount of available data is highly heterogeneous.
Our paper shows that the methods we analyze can handle such data settings as well. Specifically, we not only establish the recoverability of ranks linearly proportional to the data size but also accurately determine the precise constants governing this proportionality. Our findings indicate that these methods can handle {\it high} data heterogeneity on a much larger scale than what assumed in previous studies.

\vspace{-.0in}
\section{Preliminaries of mathematical models }
\label{sec:models}
\vspace{-.0in}

We start with the basic description of  matrix completion (MC). As it is well known (see, e.g., \cite{CR09matcomp,SAT05}), MC assumes that there exists a \emph{``ground truth''} (presumably approximately low rank) data matrix $X\in\mR^{n\times n}$ and the so-called masking matrix $M\in\mR^{n\times n}$
\begin{equation}
M_{i,j}=\begin{cases}
          1, & \mbox{component  $(i,j)$ observed}  \\
          0, & \mbox{otherwise}.
        \end{cases}.  \label{eq:mc2}
\end{equation}
The following is at the heart of matrix completion
\begin{eqnarray}\label{eq:genmcl0posmmt}
\mbox{Find} & & X \nonumber \\
\hspace{-.0in} \mbox{subject to} & & Y=M\circ X,
\end{eqnarray}
where $\circ$ stands for the component-wise multiplication  and $X$ is ``not far away'' from $X_{sol}$ in a metric of choice. Intuitively, the above means the following: one observes via mask $M$ a collection of elements of  $X_{sol}$, stores the observations in $Y$ and them, inversely, attempts to recover $X$ that fits the observations and is not far away from $X_{sol}$. To provide any notion and quantification of ``not far away'' concept the structure of $X_{sol}$ must be known and properly taken into account.

\subsection{Choice of $X_{sol}$}
\label{sec:choicex}

Various forms of $X_{sol}$'s structure can be considered. We, here, assume a scenario typically characterized as
\begin{equation*}\label{eq:eq1}
 \emph{Approximately low rank:}  \qquad  Y=M\circ X_{sol}, X_{sol}=U\mbox{diag}(\sigma+\epsilon_\sigma)V^T, \|\sigma\|_0=k,
  \end{equation*}
where $ U^TU=I_{n\times n}, V^TV=I_{n\times n}$, and the vector of singular values imperfections, $\epsilon_{\sigma}$, can be small with respect to the the vector of dominating part of singular values, $\sigma$, in a variety of metrics. Above, $I_{n\times n}$ denotes the $n\times n$ identity matrix, and $\diag(\cdot)$ the operator that creates a column vector of the diagonal entries of its matrix argument (or reversely, the diagonal matrix of the column vector entries).

The standard foundational matrix completion literature (see, e.g., \cite{CR09matcomp,SAT05,Rechtmatcomp11,CPmatcomp10}) recognizes two additional forms as of potentially equally important interest. One of them considers observations imperfections as consequences of a noisy process $\epsilon_Y$
\begin{equation*}\label{eq:eq1}
\begin{array}{lcl}
 \emph{Noisy low rank:} & \qquad & Y=M\circ X_{sol}+\epsilon_Y, X_{sol}=U\mbox{diag}(\sigma)V^T, \|\sigma\|_0=k.
 \end{array}
 \end{equation*}
The other approximates the deviations from low rankness by actually removing them and assuming the
\begin{equation*}\label{eq:eq1}
\emph{Idealized low rank sceanrio:} \qquad Y=M\circ X_{sol}, X_{sol}=U\mbox{diag}(\sigma)V^T, \|\sigma\|_0=k<n.
 \end{equation*}
The main advantage of such an idealized simplification is that it usually allows for mathematical tractability. The potential drawback is a loss of modeling and estimation accuracy as the real data rarely perfectly fit such an idealized description.


Which of the three forms is to be used heavily depends on the concrete problem at hand. For example, although the first two may look as more complete and advantageous, it is not clear that they would be the right choice in scenarios where the accuracy of the idealized one is within acceptable tolerance. Also, choosing between the first and the second needs to be done very carefully and with respect to the true source of the imperfections (deviations from the ideal low rankness) which, again, almost exclusively depends on the problem at hand. As a concrete real data example where the importance of such a choice is especially important, it is instructive to look at the California tobacco study from \cite{ADHsynth10}, which is a golden standard in econometric literature on  synthetic control. In such an example, one postulates that it is likely that over years there is a parallel trend in tobacco consumption across different states. That implies that a matrix containing annual consumption data of each state is indeed approximately low rank. This intuition is confirmed by choosing the tobacco consumption matrix that contains the annual per capita cigarette consumption in the United States and plotting its spectrum (see the Appendix). Observe that the approximate low rankness comes as a consequence of the internal structure of the data rather than as a consequence of incorrectly recorded data. This clearly indicates that the first approach should be preferable over the second. Moreover, it provides a generic conclusion, that the first option is preferable whenever the deviation from the ideal low rankness is caused by the internal structure rather than brought on externally (as in, say, incorrectly recorded observations scenarios).

It is worth adding that the first two forms can be (and have often been throughout the literature) interchangeably used as each other's emulations (particularly so as facilitating tools in heavy mathematical analyses). The analysis that we conduct is of the phase transition precision level type, which implies that the two models \textbf{\emph{cannot}} be used as replacements of each other. Moreover, given the econometrics motivating applications, the approximately low rank one has to be considered. For the concreteness, we choose $X_{sol}=U\bar{\sigma} V^T$, where $\bar{\sigma}=\diag(\sigma,\epsilon_\sigma)$ and $\sigma\in\mR^k,\epsilon_\sigma\in\mR^{n-k},\sigma_{\epsilon}\in\mR$, $\sigma\gg \sigma_\epsilon \1,\epsilon_\sigma\leq \sigma_\epsilon \1$, with $\sigma$, $\epsilon_{\sigma}$, and $\sigma_{\epsilon}$ not changing as $n$ grows ($\1$'s are (here and everywhere in the paper) vectors of all ones of appropriate  dimensions; notation $a\gg b$ means that there is a sufficiently large \emph{constant} $C$ such that $a\geq C b$).

\subsection{Choice of $M$}
\label{sec:choicex}

The problem in (\ref{eq:genmcl0posmmt}) (with a random $M$) fits within a standard MC setup. To capture the non-random block structure of missing entries (visualized in Figure \ref{fig:Mnlockcinf}), we consider the masking matrix
\begin{eqnarray}\label{eq:cinfanl2a}
  M^{(l_1,l_2)} & \triangleq & \1_{n\times 1}\1_{n\times 1}^T- I^{(l_1)}(I^{(l_1)})^T   \1_{n\times 1}\1_{n\times 1}^T   I^{(l_2)}(I^{(l_2)})^T
  \nonumber \\
   I^{(l)}  &   \triangleq  &  \begin{bmatrix}
        \0_{l\times (n-l)} \\
        I_{(n-l)\times (n-l)}
      \end{bmatrix},
\end{eqnarray}
where $\1/\0$ are the column vectors and matrices of all ones/zeros with the dimensions specified in their subscripts.
\begin{figure}[ht]
\vskip -0.17in
\begin{center}
\centerline{\includegraphics[width=13.5cm,height=7.5cm]{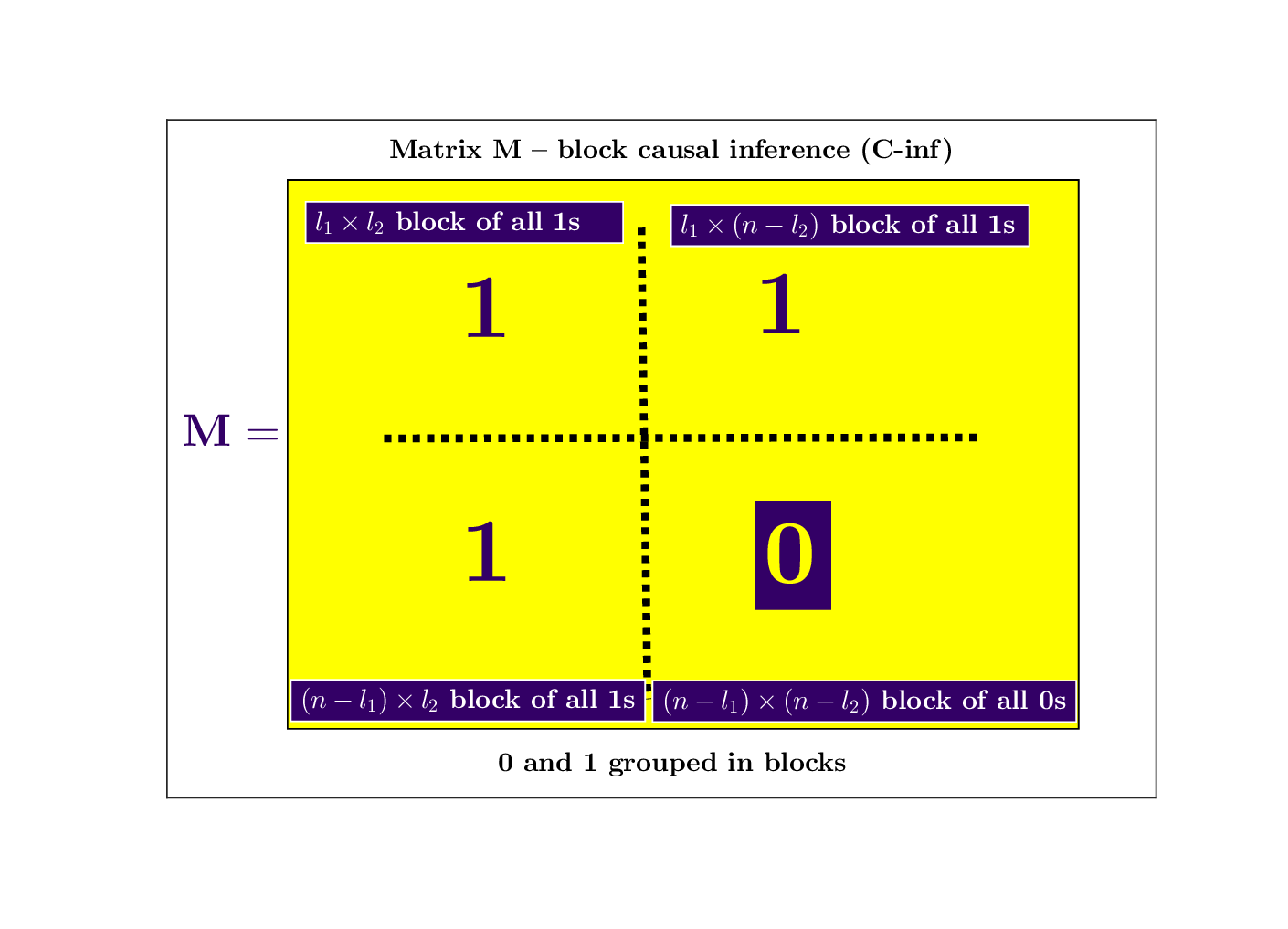}}
\vspace{-.4in}
\caption{Matrix $M\triangleq M^{(l_1,l_2)}$ -- block causal inference setup}
\label{fig:Mnlockcinf}
\end{center}
\vskip -0.32in
\end{figure}
We study the most challenging large $n$ linear regime and, for simplicity of the presentation, assume the so called square block case with $l\triangleq l_1=l_2$ and $\beta\triangleq\lim_{n\rightarrow\infty}\frac{k}{n}$ and $\eta\triangleq\lim_{n\rightarrow\infty}\frac{l}{n}$.
Such an assumption makes the analysis somewhat easier while imposing no substantial limitations on the generality of the final results. Namely, the obtained results can always be used as bounds on the non-square ones (simply trimming rectangles so that they become squares results in using less available data and obtaining larger residual $\textbf{RMSE}$'s).

\section{Mathematical analysis}\label{sec:cinfreleqv}
\vspace{-0.0in}

We start by defining $\ell_p^*(X)$ as the so-called $\ell_p$ (quasi) norm of $\sigma(X)$ (the vector of singular values of $X$),
\vspace{-0.00in}
\begin{eqnarray}\label{eq:mclinsys1b1}
\ell_p^*(X)\triangleq \ell_p(\sigma(X)), p\in\mR_+, \quad \mbox{with} \quad \ell_0^*(X) \triangleq \lim_{p\rightarrow 0}\ell_p^*(X).
\end{eqnarray}\vspace{-0.0in}
Instead of solving (\ref{eq:genmcl0posmmt}), the MC literature often opts for its tightest convex relaxation
\vspace{-0.0in}
\begin{eqnarray}\label{eq:genmcl1posmmt}
			\hat{X} \triangleq  \min_{X} & & \ell_1^*(X) \nonumber \\
			\mbox{subject to} & & Y=M\circ X, Y\triangleq M\circ X_{sol}. \hspace{.0in}
			\vspace{-.0in}
\end{eqnarray}
We refer to the introductory papers \cite{RFPrank,CR09matcomp,SAT05,FazHinBoyd01} for additional background on the relation between $\ell_0^*$ and $\ell_1^*$ and connections to low rank recovery (LRR) and MC (related considerations, can be found in a long line of excellent follow-up work, e.g., \cite{CandesRecht11simple,CT10matcomp,CanLiMaWri09,KMO10matcomp,KMO10matcomp1,Klopp14matcomp,KLT11matcomp,NW11matcomp,NW12matcomp,RT11matcomp,MHT10}). It is well known in the structured recovery literature that, in the idealized scenario, $\sigma_\epsilon=0$, the above heuristics often produce the true solution, i.e., one often has $\hat{X}=X_{sol}$. What is somewhat hidden in such a folklore knowledge, is that the masking matrix $M$ is typically viewed as well balanced or comprised of randomly located zeros and ones. The problem we study, though, presents three key differences: 1) scenario of interest is \emph{non-idealized}, i.e., $\sigma_\epsilon>0$; 2) $M$ is not random; and 3) the locations of zeros and ones form a very particular \emph{block} structure. These differences make it impossible to follow standard/known CS/MC analytical paths. Instead, we take the differences to our advantage, avoid the known analytical paths, and create a new one to characterize the $\ell_1^*$'s performance. Since the considered scenario is non-idealized, there will be a nonzero residual estimation error $\hat{W}\triangleq\hat{X}-X_{sol}$. Any $\ell_1^*$ performance analysis then must account for a characterization of $\hat{W}$. Our approach is a no exception in that sense. However, as particularly relevant to us, we highlight below the \emph{precision} of the conducted analysis.
\vspace{-.0in}
\begin{center}
 \tcbset{beamer,lower separated=false, fonttitle=\bfseries,
coltext=black , interior style={top color=orange!10!yellow!30!white, bottom color=yellow!80!yellow!50!white}, title style={left color=black, right color=red!70!orange!30!blue}}
\begin{tcolorbox}[title=Ultimate performance analysis goal:, width=5.2in]
\vspace{-.18in}
 \begin{eqnarray*}\label{eq:genmcl1poseqvrmse}
\mbox{\underline{\textbf{\emph{Precise}}} characterization of} \hspace{.03in} \mathbf{RMSE}_{(ns)}\triangleq \|\hat{W}\|_F \triangleq \|\hat{X}-X_{sol}\|_F.
\end{eqnarray*}
 \vspace{-.3in}
\end{tcolorbox}
\end{center}
\vspace{-.05in}
Below, we provide such a characterization relying on typical statistical viewing of the problem 
and utilizing advanced sophisticated probabilistic concepts that we introduce along the way. We begin with the following algebraic result.
\begin{theorem}(Algebraic characterization of $W$)
Consider a $\bU\in\mR^{n\times k}$ such that $\bU^T\bU=I_{k\times k}$ and a $\bV\in\mR^{n\times k}$ such that $\bV^T\bV=I_{k\times k}$ and an approximately  rank $k$ matrix $X_{sol}=X\in\mR^{n\times n}$, such that $X_{sol}=U\bar{\sigma}V^T, \bar{\sigma}=\diag(\sigma,\epsilon_\sigma),\sigma\in\mR^k,\epsilon_\sigma\in\mR^{n-k},\sigma\gg \sigma_\epsilon,\epsilon_\sigma\leq \sigma_\epsilon$.
Also, let the orthogonal spans $\bU^{\perp}\in\mR^{n\times (n-k)}$ and $\bV^{\perp}\in\mR^{n\times (n-k)}$ be such that $U\triangleq \begin{bmatrix}
    \bU & \bU^{\perp}
   \end{bmatrix}$ and $V\triangleq \begin{bmatrix}
    \bV & \bV^{\perp}
   \end{bmatrix}$ and
\begin{equation}\label{eq:cinfthm0}
U^TU\triangleq \begin{bmatrix}
    \bU & \bU^{\perp}
   \end{bmatrix}^T\begin{bmatrix}
    \bU & \bU^{\perp}
   \end{bmatrix}=I_{n\times n} \quad \mbox{and} \quad
 V^TV \triangleq\begin{bmatrix}
    \bV & \bV^{\perp}
   \end{bmatrix}^T\begin{bmatrix}
    \bV & \bV^{\perp}
   \end{bmatrix}=I_{n\times n}.
 \end{equation}
With $M\in\mR^{n\times n}$ as in (\ref{eq:cinfanl2a}), assume that $Y= M\circ X_{sol}$ and let $\hat{X}$ be the solution of (\ref{eq:genmcl1posmmt}). Set
 \begin{equation}
\hat{W} = \arg \min_{W,M\circ W=\0_{n\times n}}  \tr(\bU^TW\bV)+\ell_1^*(\diag(\epsilon_\sigma)+(\bU^{\perp})^TW\bV^{\perp}).
\label{eq:cinfcor1}
\end{equation}
Then for any $X_{sol}$ that satisfies the above setup one has
\begin{equation}
 \textbf{\emph{RMSE}}_{(ns)}=\|\hat{X}-X_{sol}\|_F\leq\|\hat{W}\|_F.\label{eq:cinfcor1a}
\end{equation}
Moreover, there exists an $X_{sol}$ such that
\begin{equation}
 \textbf{\emph{RMSE}}_{(ns)}=\|\hat{X}-X_{sol}\|_F=\|\hat{W}\|_F.\label{eq:cinfcor1a1}
\end{equation}
   \label{thm:thm1}
\end{theorem}
\begin{proof}
Presented in the Appendix.
\end{proof}

To analyze the optimization problem presented in (\ref{eq:cinfcor1}), we adopt an approach that relies on using a generic Lagrangian mechanism. The initial step in this process involves explicitly formulating the optimization problem
 from (\ref{eq:cinfcor1})
\begin{eqnarray}
f_{pr}(M;U,V) \triangleq \min_{W} & &  \tr(\bU^TW\bV)+\ell_1^*(\diag(\epsilon_\sigma)+(\bU^{\perp})^TW\bV^{\perp}) \nonumber \\
\mbox{subject to} & & M\circ W=\0_{n\times n}.
\label{eq:cinfanl2}
\end{eqnarray}
One can then write the corresponding Lagrangian and the Lagrange dual function,
\begin{eqnarray}
{\cal L}(W,\Lambda) & \triangleq &  \tr(\bU^TW\bV)+\ell_1^*(\diag(\epsilon_\sigma)+(\bU^{\perp})^TW\bV^{\perp})+\Theta (M\circ W) \nonumber \\
& = & \max_{\Lambda,\Lambda^T\Lambda\leq I} \lp \tr(\bU^TW\bV)+\tr(\Lambda(\diag(\epsilon_\sigma)+(\bU^{\perp})^TW\bV^{\perp}))+\Theta (M\circ W)\rp\nonumber \\
& = & \max_{\Lambda,\Lambda^T\Lambda\leq I} \lp \tr\lp (\bV\bU^T+\bV^{\perp}\Lambda(\bU^{\perp})^T+\Theta\circ M)W\rp + \tr(\Lambda\diag(\epsilon_\sigma)) \rp,\label{eq:cinfanl3}
\end{eqnarray}
and
\begin{eqnarray}
g(\Theta) & \triangleq & \min_{W} {\cal L}(W,\Lambda,\Theta,\gamma) \nonumber \\
 & = & \min_{W}\max_{\Theta,\Lambda,\Lambda^T\Lambda\leq I} \lp \tr\lp (\bV\bU^T+\bV^{\perp}\Lambda(\bU^{\perp})^T+\Theta\circ M)W\rp + \tr(\Lambda\diag(\epsilon_\sigma) \rp.\label{eq:cinfanl4}
\end{eqnarray}
Utilizing the Lagrangian duality one then further has
\begin{eqnarray}
g(\Theta) & = & \min_{W}\max_{\Lambda,\Lambda^T\Lambda\leq I} \lp \tr\lp (\bV\bU^T+\bV^{\perp}\Lambda(\bU^{\perp})^T+\Theta\circ M)W\rp + \tr(\Lambda\diag(\epsilon_\sigma) \rp \nonumber \\
& \geq & \max_{\Lambda,\Lambda^T\Lambda\leq I}  \min_{W}\lp \tr\lp (\bV\bU^T+\bV^{\perp}\Lambda(\bU^{\perp})^T+\Theta\circ M)W\rp + \tr(\Lambda\diag(\epsilon_\sigma) \rp.\label{eq:cinfanl5}
\end{eqnarray}
Moreover, using the definition of $M$ from (\ref{eq:cinfanl2a}) we have
 \begin{eqnarray}
f_{pr}(M;U,V) & \geq  & \max_{\Theta} g(\Theta) \nonumber \\
 & = & \max_{\Lambda,\Lambda^T\Lambda\leq I,(I^{(l)})^T(\bV\bU^T+\bV^{\perp}\Lambda(\bU^{\perp})^TI^{(l)}=0} \tr(\Lambda\diag(\epsilon_\sigma)).\label{eq:cinfanl5a}
\end{eqnarray}
It is now critically important to observe that the above holds as long as the following optimizing condition is feasible
			\begin{equation}
			 \exists \Lambda| \Lambda^T\Lambda\leq I \quad \mbox{and} \quad (I^{(l)})^T \lp \bV\bU^T+\bV^{\perp}\Lambda(\bU^{\perp})^T\rp I^{(l)}=0.\label{eq:cinfcor2eq1}
			\end{equation}
On the other hand, as soon as the system parameters are such that it becomes infeasible, one has that both optimization in (\ref{eq:cinfanl5})  and the residual error norm, $\|W\|_F$, are unbounded. Clearly, to be able to conduct the analysis of the residual error, one must first determine the range or parameters where such error is bounded, i.e., where the  optimizing condition in (\ref{eq:cinfanl5a}) (or (\ref{eq:cinfcor2eq1})) is feasible.

\subsection{Feasibility of the optimizing condition in (\ref{eq:cinfanl5a})}
\label{sec:feasoptcond}

For the time being we assume $k\leq l$ (later on this assumption will be rigorously justified). From (\ref{eq:cinfcor2eq1}), one then easily has
{\small \begin{equation}
	\Lambda=((I^{(l)})^T \bV^{\perp})^{-1} (I^{(l)})^T \bV \bU^T I^{(l)}((\bU^{\perp})^T I^{(l)})^{-1}  \Longrightarrow (I^{(l)})^T \lp \bV\bU^T+\bV^{\perp}\Lambda(\bU^{\perp})^T\rp I^{(l)}=0, \label{eq:cinfanl8}
\end{equation}}
where $(\cdot)^{-1}$ stands for the pseudo-inverse. To make writing a bit easier we can set
\begin{eqnarray}
	\Lambda_{opt} & \triangleq & \Lambda_V\Lambda_U^T \nonumber \\
	\Lambda_V & \triangleq & ((I^{(l)})^T \bV^{\perp})^{-1} (I^{(l)})^T \bV \nonumber \\
	\Lambda_U & \triangleq & ((I^{(l)})^T \bU^{\perp})^{-1} (I^{(l)})^T \bU. \label{eq:cinfanl9}
\end{eqnarray}
Let $\lambda_{max}(\cdot)$ be the maximum eigenvalue of its symmetric matrix argument. After combining (\ref{eq:cinfcor2eq1}) and Theorem \ref{thm:thm1}, we conclude that if
\begin{eqnarray}
	\lambda_{max}(\Lambda_{opt}^T\Lambda_{opt}) \leq 1, \label{eq:cinfanl10}
\end{eqnarray}
then (\ref{eq:cinfcor2eq1}) will be satisfied. After basic algebraic transformations, (\ref{eq:cinfanl10}) can also be rewritten as
\begin{equation}
	\lambda_{max}(\Lambda_{opt}^T\Lambda_{opt})=\lambda_{max}(\Lambda_{opt}\Lambda_{opt}^T)=\lambda_{max}(\Lambda_V\Lambda_U^T\Lambda_U\Lambda_V^T)
	=\lambda_{max}(\Lambda_V^T\Lambda_V\Lambda_U^T\Lambda_U)\leq 1. \label{eq:cinfanl11}
\end{equation}
From (\ref{eq:cinfanl11}), it is clear that the spectrum of $\Lambda_V^T\Lambda_V\Lambda_U^T\Lambda_U$ as well as the spectra of $\Lambda_V^T\Lambda_V$ and $\Lambda_U^T\Lambda_U$ play an important role in the $\ell_0^*-\ell_1^*$-equivalence. We first observe a \emph{worst case} bound. Namely, since
\begin{eqnarray}
	\lambda_{max}(\Lambda_{opt}^T\Lambda_{opt})=\lambda_{max}(\Lambda_V^T\Lambda_V\Lambda_U^T\Lambda_U)
	\leq \lambda_{max}(\Lambda_V^T\Lambda_V)\lambda_{max}(\Lambda_U^T\Lambda_U), \label{eq:cinfanl12}
\end{eqnarray}
one has that if the individual spectra of $\Lambda_V^T\Lambda_V$ and $\Lambda_U^T\Lambda_U$ do not exceed one then the $\ell_0^*-\ell_1^*$-equivalence holds. Given the obvious importance of these spectra, we below look at them in more detail. Clearly, due to symmetry, it is sufficient to focus on only one of them. To that end we start by observing
\begin{eqnarray}
	((I^{(l)})^T \bV^{\perp})^{-1} & = & (\bV^{\perp})^T I^{(l)} \lp (I^{(l)})^T\bV^{\perp}(\bV^{\perp})^T I^{(l)}\rp^{-1}. \label{eq:cinfanl13}
\end{eqnarray}
From (\ref{eq:cinfanl13}), one quickly finds
\begin{eqnarray}
	\lp ((I^{(l)})^T \bV^{\perp})^{-1} \rp^T ((I^{(l)})^T \bV^{\perp})^{-1} & = & \lp (I^{(l)})^T\bV^{\perp}(\bV^{\perp})^T I^{(l)}\rp^{-1}. \label{eq:cinfanl14}
\end{eqnarray}
We can then write
\begin{eqnarray}
	(I^{(l)})^T \bV \bV^T I^{(l)} & = &  (I^{(l)})^T \lp I - \bV^{\perp} (\bV^{\perp})^T \rp I^{(l)}. \label{eq:cinfanl15}
\end{eqnarray}
From (\ref{eq:cinfanl9}), we also have
\begin{eqnarray}
	Q_1\triangleq \Lambda_V^T\Lambda_V & = &  \lp \lp ((I^{(l)})^T \bV^{\perp})^{-1} (I^{(l)})^T \bV\rp^T \lp ((I^{(l)})^T \bV^{\perp})^{-1} (I^{(l)})^T \bV \rp \rp \nonumber \\
	& = &    \bV^T I^{(l)}\lp ((I^{(l)})^T \bV^{\perp})^{-1} \rp^T ((I^{(l)})^T \bV^{\perp})^{-1}   \lp(I^{(l)})^T \bV \rp.\label{eq:cinfanl16}
\end{eqnarray}
Now, we find it more convenient to work with a slightly changed version of matrix $Q_1$. Namely, after combining (\ref{eq:cinfanl9}), (\ref{eq:cinfanl14}), and (\ref{eq:cinfanl15}), we obtain
\begin{eqnarray}
	Q & \triangleq &  \lp \lp ((I^{(l)})^T \bV^{\perp})^{-1} \rp^T ((I^{(l)})^T \bV^{\perp})^{-1}   \lp(I^{(l)})^T \bV \bV^T I^{(l)}\rp \rp \nonumber \\
	& = &   \lp \lp (I^{(l)})^T\bV^{\perp}(\bV^{\perp})^T I^{(l)}\rp^{-1}  \lp (I^{(l)})^T \lp I - \bV^{\perp} (\bV^{\perp})^T \rp I^{(l)}\rp \rp\nonumber \\
	& = &   \lp (I^{(l)})^T\bV^{\perp}(\bV^{\perp})^T I^{(l)}\rp^{-1} -I. \label{eq:cinfanl16a0}
\end{eqnarray}
Clearly, all the nonzero eigenvalues of $Q_1$ and $Q$ are identical. When $k\leq n-l$ then $Q$ has all the eigenvalues of $Q_1$ plus $n-l-k$ extra zeros. On the other hand, when $k\geq n-l$ then $Q_1$ has all the eigenvalues of $Q$ plus $k-(n-l)$ extra zeros. Since adding or removing zeros from the spectra does not change any of their features of our interests here, instead of working directly with $Q_1$, we can work with $Q$. In particular, we have
\begin{equation}
	\lambda_{max}(Q_1) = \lambda_{max}(\Lambda_V^T\Lambda_V)
	 =  \lambda_{max} \lp \lp (I^{(l)})^T\bV^{\perp}(\bV^{\perp})^T I^{(l)}\rp^{-1}\rp -1 = \lambda_{max}(Q). \label{eq:cinfanl16a}
\end{equation}

We are now in position to establish a spectral alternative to Corollary 2 from \cite{CapStoAnnalsStatsSupp23}.

\begin{proposition}(\textbf{Idealized low rank scenario --- \bl{$\ell_0^*-\ell_1^*$-equivalence condition via mask-modified bases spectra}})
	Assume the setup of Proposition 1 and Corollaries 1 and 2 from \cite{CapStoAnnalsStatsSupp23} with $k\leq l$. Let $\lambda_V$ and $\lambda_U$ be defined as in (\ref{eq:cinfanl9}). Then
	\begin{center}
\tcbset{beamer,lower separated=false, fonttitle=\bfseries,
coltext=black , interior style={top color=orange!10!yellow!30!white, bottom color=yellow!80!yellow!50!white}, title style={left color=orange!10!cyan!30!blue, right color=green!70!blue!20!black}}
		\begin{tcolorbox}[beamer,title= \yellow{$\ell_0^*\Longleftrightarrow \ell_1^* \quad \mbox{and}\quad  \textbf{\emph{RMSE}}=\|\hat{X}-X_{sol}\|_F=0$},fonttitle=\bfseries,width=.7\linewidth]
			\begin{equation}
				\mbox{If and only if} \quad \lambda_{max}(\Lambda_V^T\Lambda_V\Lambda_U^T\Lambda_U)\leq 1.\label{eq:cinfcor3eq1}
			\end{equation}
		\end{tcolorbox}
	\end{center}
	Moreover, if
	\begin{equation}
		\lp\lambda_{max} \lp \lp (I^{(l)})^T\bV^{\perp}(\bV^{\perp})^T I^{(l)}\rp^{-1}\rp -1\rp
		\lp\lambda_{max} \lp \lp (I^{(l)})^T\bU^{\perp}(\bU^{\perp})^T I^{(l)}\rp^{-1}\rp -1\rp \leq 1,\label{eq:cinfcor3eq2}
	\end{equation}
	then again $\ell_0^*\Longleftrightarrow \ell_1^*$ and $\textbf{\emph{RMSE}}=\|\hat{X}-X_{sol}\|_F=0$.
	\label{cor:cinfcor3}
\end{proposition}
\begin{proof}
	The first part follows from Corollaries 1 and 2 from \cite{CapStoAnnalsStatsSupp23}, (\ref{eq:cinfanl8}), (\ref{eq:cinfanl9}), (\ref{eq:cinfanl11}), the above discussion and some additional considerations while the second part follows by additionally taking into account (\ref{eq:cinfanl12}) and (\ref{eq:cinfanl16a}). We below present all the details split into three parts: the first two relate to the equivalence condition (equation (\ref{eq:cinfcor3eq1})) while the third one relates to (\ref{eq:cinfcor3eq2}).
	
	\noindent \underline{\bl{\textbf{\emph{1) $\Longrightarrow$ -- The ``if part" of condition (\ref{eq:cinfcor3eq1}):}}}} Choosing $\Lambda=\Lambda_{opt}$
	\begin{eqnarray}
		\hspace{-.0in}\Lambda_{opt} & \triangleq & \Lambda_V\Lambda_U^T  =  -((I^{(l)})^T \bV^{\perp})^{-1} (I^{(l)})^T \bV \bU^T I^{(l)}((\bU^{\perp})^T I^{(l)})^{-1}, \label{eq:cinfanl8z1}
	\end{eqnarray}
	(where, we recall that, $(\cdot)^{-1}$ stands for the pseudo-inverse) one ensures
	\begin{equation} \label{eq:cinfanl9z1}
		(I^{(l)})^T \lp \bV\bU^T+\bV^{\perp}\Lambda(\bU^{\perp})^T\rp I^{(l)}=0.
	\end{equation}
	Let $\lambda_{max}(\cdot)$ be the maximum eigenvalue of its symmetric matrix argument. A combination of (\ref{eq:cinfcor2eq1})-(\ref{eq:cinfanl9z1}) ensures that if
	\begin{eqnarray}
		\lambda_{max}(\Lambda_{opt}^T\Lambda_{opt}) \leq 1, \label{eq:cinfanl10z1}
	\end{eqnarray}
	then $\Lambda_{opt}$ satisfies (\ref{eq:cinfcor2eq1}). (\ref{eq:cinfanl10z1}) is implied by (\ref{eq:cinfcor3eq1}) since
	\begin{equation}
		\lambda_{max}(\Lambda_{opt}^T\Lambda_{opt})  =  \lambda_{max}(\Lambda_U\Lambda_V^T\Lambda_V\Lambda_U^T)  = \lambda_{max}(\Lambda_V^T\Lambda_V\Lambda_U^T\Lambda_U)  \leq 1, \nonumber \label{eq:cinfanl11z1}
	\end{equation}
	which suffices to complete the proof of the ``if part".
	
	\noindent \underline{\bl{\textbf{\emph{2) $\Longleftarrow$ -- The ``only if part" of condition (\ref{eq:cinfcor3eq1}):}}}}  Consider SVDs
	\begin{equation}\label{eq:cinfanl11aa1z1}
		B \triangleq (I^{(l)})^TV^\perp=U_B\Sigma_BV_B^T,\quad C \triangleq (I^{(l)})^TU^\perp=U_C\Sigma_CV_C^T
	\end{equation}
	with unitary $U_B,V_B,U_C,V_C$ and diagonal (with no zeros on the main diagonal) $\Sigma_B,\Sigma_C$. Any $\Lambda$ can be parameterized as
	\begin{equation}\label{eq:cinfanl11aa2z1}
		\Lambda = V_BH^T+V_B^\perp D^T,\quad H\triangleq V_C E+V_C^\perp F,
	\end{equation}
	for some $E,F,D$ and unitary $V_B^\perp$ and $V_C^\perp$ such that $V_B^TV_B^\perp=V_C^TV_C^\perp=0$. Also, one can set $\Lambda_*$ and write the SVD of $E$
	\begin{equation}\label{eq:cinfanl11aa3z1}
		\Lambda_* \triangleq V_B E^TV_C^T, \quad E=U_E\Sigma_EV_E^T,
	\end{equation}
	where $U_E,V_E$ are unitary and $\Sigma_E$ is diagonal with entries on the main diagonal being nonzero and in the ascending order. Let $\u_e$ be the last column of $U_E$ (i.e. the eigenvector of $EE^T$ that corresponds to its largest eigenvalue). Since $\|V_C\u_e\|_2=1$,
	\begin{eqnarray}\label{eq:cinfanl11aa4z1}
		\lambda_{max}(\Lambda^T\Lambda)& \geq  &   \u_e^TV_C^T \Lambda^T\Lambda V_C\u_e \nonumber \\
		& = & \u_e^TV_C^T HH^T V_C\u_e  +\u_e^TV_C^T DD^T V_C\u_e  \nonumber \\
		& \geq  & \u_e^TV_C^T (V_C E+V_C^\perp F)(V_C E+V_C^\perp F)^T V_C\u_e    \nonumber \\
		& =  & \u_e^T EE^T\u_e =\lambda_{max}(EE^T) \nonumber \\
		& = &  \lambda_{max}(V_CEE^TV_C^T) =\lambda_{max}(\Lambda_*^T\Lambda_*).
	\end{eqnarray}
	If $\Lambda$ satisfies the condition of (\ref{eq:cinfcor2eq1}) then a combination of (\ref{eq:cinfcor2eq1}) and (\ref{eq:cinfanl11aa1z1})-(\ref{eq:cinfanl11aa3z1}) gives
	\begin{equation}\label{eq:cinfanl11aa5z1}
		(I^{(l)})^T\bar{V}\bar{U}I^{(l)}+B\Lambda_*C^T=0,
	\end{equation}
	and a combination of (\ref{eq:cinfanl9}), (\ref{eq:cinfanl11aa1z1}), and (\ref{eq:cinfanl11aa5z1}) gives
	\begin{equation}\label{eq:cinfanl11aa6z1}
		\Lambda_V\Lambda_U^T  =   -B^{-1}(I^{(l)})^T\bar{V}\bar{U}I^{(l)}(C^T)^{-1} = \Lambda_*.
	\end{equation}
	Finally, for  $\Lambda$ that fits (\ref{eq:cinfcor2eq1}), from (\ref{eq:cinfanl11aa4z1}) and (\ref{eq:cinfanl11aa6z1}) one has
	\begin{eqnarray}\label{eq:cinfanl11aa7z1}
		1 & \geq & \lambda_{max}(\Lambda^T\Lambda)  >  \lambda_{max}(\Lambda_*^T\Lambda_*)=\lambda_{max}(\Lambda_*\Lambda_*^T) \nonumber \\
		& = & \lambda_{max}(\Lambda_V\Lambda_U^T\Lambda_U\Lambda_V^T) = \lambda_{max}(\Lambda_V^T\Lambda_V\Lambda_U^T\Lambda_U),
	\end{eqnarray}
	which completes the proof of the ``only if part".

	\noindent \underline{\bl{\textbf{\emph{3) Suffciency of the condition (\ref{eq:cinfcor3eq2}):}}}}  Since
	\begin{equation}
		\lambda_{max}(\Lambda_V^T\Lambda_V\Lambda_U^T\Lambda_U)
		\leq \lambda_{max}(\Lambda_V^T\Lambda_V)\lambda_{max}(\Lambda_U^T\Lambda_U), \label{eq:cinfanl12z1}
	\end{equation}
	one has that if the individual spectra of $\Lambda_V^T\Lambda_V$ and $\Lambda_U^T\Lambda_U$ do not exceed one then the $\ell_0^*-\ell_1^*$-equivalence holds. Repeating (\ref{eq:cinfanl13})-(\ref{eq:cinfanl16a0}), with $V$ replaced by $U$, $Q_1$ by $Q_1^{\perp}$, and $Q$ by $Q^{\perp}$ one arrives at the following analogue of (\ref{eq:cinfanl16a})
	\begin{equation}
		\lambda_{max}(\Lambda_U^T\Lambda_U) = \lambda_{max}(Q_1^{\perp})= \lambda_{max}(Q^{\perp}). \label{eq:cinfanl16a2z1}
	\end{equation}
	A combination of (\ref{eq:cinfanl12z1}) and (\ref{eq:cinfanl16a}) - (\ref{eq:cinfanl16a2z1}) completes the proof of the condition (\ref{eq:cinfcor3eq2})'s sufficiency  for the $\ell_0^*-\ell_1^*$-equivalence.
\end{proof}

The feasibility of the optimizing condition is precisely what determines the so-called phase transition in the idealized scenario where $\epsilon_\sigma=0$. The above machinery is actually powerful enough to precisely determine, in a statistical context, the location of the PT curve. We assume a generic \emph{typical} statistical scenario, where the orthogonal spans of interests are Haar distributed (this is a bit less restrictive than the usual i.i.d. Gaussian assumptions from, say, factor models literature, see, e.g., \cite{XPlatfac10}, and it does not impose \emph{per se} the so-called \emph{strong} factor/loading assumptions see, e.g. \cite{BaiNg19,CahBaiNg21} (note that Gaussians are only a special case of the Haar distributed objects)).\footnote{The assumed rotational invariance is effectively analogous to the assumption of randomly positioning the sparse portion within the unknown vectors in basic compressed sensing (with the exception of results of  \cite{DonohoPol}, this assumption is almost always in place in the theoretical compressed sensing \emph{phase transition} characterizations).} In the absence of any \emph{a priori} statistically available knowledge about the data, the above is the minimum that realistically \textbf{\emph{must}} be assumed. In other words, if, beyond the degree of heterogeneity, nothing more is known about the data, no modeling concept can be proven as better than \emph{uniform} randomness.

We first need the following lemma.

 \begin{lemma}\label{lemma:lemma1}
Assume large $n$ linear regime with
\begin{eqnarray}
\beta\triangleq \lim_{n\rightarrow \infty}\frac{k}{n}\quad \mbox{and} \quad \eta\triangleq \lim_{n\rightarrow \infty}\frac{l}{n}. \label{eq:typwclemma2eq1}
\end{eqnarray}
Let $\bV^{\perp}\in\mR^{n\times (n-k)}$ be a Haar distributed basis of an $n-k$-dimensional subspace of $R^n$. Analogously, let $\bU_D^{\perp}\in\mR^{n\times (n-k)}$ be a Haar distributed basis of an $n-l$-dimensional subspace of $R^n$. Moreover, let $\bV^{\perp}\in\mR^{n\times (n-k)}$ and $\bU_D^{\perp}\in\mR^{n\times (n-k)}$ be independent of each other. Also, let ${\cal V}$, ${\cal U}$, and $\tilde{D}$ be defined as
\begin{equation}
{\cal V}  \triangleq \bV^{\perp}(\bV^{\perp})^T, \quad {\cal U}  \triangleq  \bU_D^{\perp}(\bU_D^{\perp})^T,\quad
 \tilde{D} \triangleq  {\cal V}{\cal U}. \label{eq:typwclemma2eq2}
\end{equation}
Set $x_l$ and $x_u$ as
\begin{eqnarray}
 x_l & \triangleq & \beta+\eta-2\beta\eta  -  \sqrt{(\beta+\eta-2\beta\eta)^2-(\beta-\eta)^2} \nonumber \\
 x_u & \triangleq & \beta+\eta-2\beta\eta  +  \sqrt{(\beta+\eta-2\beta\eta)^2-(\beta-\eta)^2},
 \label{eq:typwclemma2eq3}
\end{eqnarray}
Then the limiting spectral distribution of $\tilde{D}$, $f_{\tilde{D}}(x)$, is
\begin{eqnarray}
     f_{\tilde{D}}(x)   &  =  & f_0\delta(x) +  f_{\tilde{D}}^{(b)}(x) + f_1f_0\delta(x-1) \nonumber \\
     & = &
              \max(\beta,\eta)\delta(x) + f_{\tilde{D}}^{(b)}(x)
              +\max(1-(\beta+\eta),0)\delta(x-1),
     \label{eq:typwclemma2eq4}
\end{eqnarray}
with
\begin{eqnarray}
     f_{\tilde{D}}^{(b)}(x)   &  =  &
            \begin{cases}
              \frac{\sqrt{-(x-(\beta+\eta))^2-4\beta\eta(x-1)}}{2\pi(x-x^2)}, & \mbox{if }  x_l\leq x\leq x_u. \\
              0, & \mbox{otherwise}.
            \end{cases}
     \label{eq:typwclemma2eq5}
\end{eqnarray}
 \end{lemma}
\begin{proof}
Presented in the Appendix.
\end{proof}

 Now we write
\begin{equation}\label{eq:typwcanl80}
  \ell_0^*-\ell_1^*-\mbox{equivalence} \quad\Longleftrightarrow \quad \frac{1}{x_l}-1\leq 1.
 \end{equation}
 Recalling on (\ref{eq:typwclemma2eq3})
\begin{eqnarray}
  x_l & \triangleq & \beta+\eta-2\beta\eta  -  \sqrt{(\beta+\eta-2\beta\eta)^2-(\beta-\eta)^2},
   \label{eq:typwcanl81}
\end{eqnarray}
 and combining further with (\ref{eq:typwcanl80}), we obtain
\begin{eqnarray}
\begin{array}{rrcl}
 & \frac{1}{2} & \leq & x_l \\
\Longleftrightarrow & \frac{1}{2} & = & \beta+\eta-2\beta\eta  -  \sqrt{(\beta+\eta-2\beta\eta)^2-(\beta-\eta)^2} \\
\Longleftrightarrow &  (\beta+\eta-2\beta\eta)^2-(\beta-\eta)^2 & \leq & \lp\beta+\eta-2\beta\eta  - \frac{1}{2}\rp^2 \\
\Longleftrightarrow &   (\beta+\eta-2\beta\eta)^2-(\beta-\eta)^2 & \leq & \lp\beta+\eta-2\beta\eta \rp^2 - \lp\beta+\eta-2\beta\eta \rp + \frac{1}{4} \\
\Longleftrightarrow &    -(\beta-\eta)^2  & \leq & - \lp\beta+\eta-2\beta\eta \rp + \frac{1}{4} \\
\Longleftrightarrow &  0 & \leq & \beta^2+\eta^2 - \lp\beta+\eta \rp + \frac{1}{4}  \\
\Longleftrightarrow &  \eta-\eta^2  & \leq &  \lp \frac{1}{2} - \beta \rp^2\\
\Longleftrightarrow &   \beta & \leq & \frac{1}{2} - \sqrt{\eta-\eta^2}.
\end{array}
  \label{eq:typwcanl82}
\end{eqnarray}
From (\ref{eq:typwcanl80}) and (\ref{eq:typwcanl82}), we finally have
\begin{equation}\label{eq:typwcanl83}
  \ell_0^*-\ell_1^*-\mbox{equivalence} \quad\Longleftrightarrow \quad \beta  \leq  \frac{1}{2} - \sqrt{\eta-\eta^2}.
 \end{equation}

We are now in position to formalize the above results and establish the so-called phase-transition phenomenon as well as its a precise \emph{worst case} location in a \emph{typical} statistical scenario.

\begin{theorem}(\textbf{\bl{$\ell_1^*$ -- phase transition -- (typical \underline{worst case})}})
  Consider a rank-$k$ matrix  $X_{sol}=X\in\mR^{n\times n}$ with the Haar distributed (\emph{not necessarily independent}) bases of its orthogonal row and column spans $\bU^{\perp}\in\mR^{n\times (n-k)}$ and $\bV^{\perp}\in\mR^{n\times (n-k)}$ ($X_{sol}^T\bU^{\perp}=X_{sol}\bV^{\perp}=\0_{n\times (n-k)}$). Let $M\triangleq M^{(l)}\in\mR^{n\times n}$ be as defined in (\ref{eq:cinfanl2a}). Assume a large $n$ linear regime with $\beta\triangleq\lim_{n\rightarrow\infty}\frac{k}{n}$ and $\eta\triangleq\lim_{n\rightarrow\infty}\frac{l}{n}$ and let $\beta_{wc}$ and $\eta$ satisfy the  following
\begin{center}
\tcbset{beamer,lower separated=false, fonttitle=\bfseries,
coltext=black , interior style={top color=orange!10!yellow!30!white, bottom color=yellow!80!yellow!50!white}, title style={left color=orange!10!cyan!30!blue, right color=green!70!blue!20!black}}
 \begin{tcolorbox}[beamer,title=\textbf{$\ell_1^*$ \yellow{worst case} phase transition (PT)} characterization,lower separated=false,,fonttitle=\bfseries,width=5in]
  \begin{equation}\label{eq:typwcthm1eq1}
    \xi_{\eta}^{(wc)}(\beta) \triangleq \beta-\frac{1}{2}+\sqrt{\eta-\eta^2}=0.
  \end{equation}
 \end{tcolorbox}
\end{center}
 \noindent \textbf{If and only if} $\beta\leq \beta_{wc}$
    \begin{equation}\label{eq:typwcthm1eq2}
   \lim_{n\rightarrow\infty} \mP(\ell_0^*\Longleftrightarrow \ell_1^*)=\ \lim_{n\rightarrow\infty} \mP(\mathbf{RMSE}=0)=1,
  \end{equation}
and the solutions of (\ref{eq:genmcl0posmmt}) and (\ref{eq:genmcl1posmmt}) coincide with overwhelming probability.
  \label{thm:typwcthm1}
\end{theorem}
\begin{proof}
The ``if'' part follows from Proposition \ref{cor:cinfcor3}, Lemma \ref{lemma:lemma1}, and the above discussion. The ``only if'' part additionally assumes $\bU^{\perp}=\bV^{\perp}$ which, as discussed above, ensures that the results are, in the worst case, achievable.
\end{proof}

We visually illustrate the results from the above theorem in Figure \ref{fig:cinfspectypwcPTbetaeta}. As can be seen from the figure, the phase transition curve splits the entire $(\beta,\eta)$ region into two subregions. The first of the subregions is below (or to the right of) the curve and in that region the $\ell_0^*-\ell_1^*$-equivalence phenomenon occurs. This means that one can recover an ideally low rank $X_{sol}$ masked by $M$ as in (\ref{eq:genmcl0posmmt}) via the $\ell_1^*$ heuristic from (\ref{eq:genmcl1posmmt}) with the residual mean square error ($\textbf{RMSE}$) equal to zero. In other words, for the system parameters $(\beta,\eta)$ that belong to the subregion below the curve one has a perfect recovery with $X_{sol}$ and $\hat{X}$ (the respective solutions of (\ref{eq:genmcl0posmmt}) and (\ref{eq:genmcl1posmmt})) being equal to each other and consequently with $\textbf{RMSE}=\|\hat{X}-X_{sol}\|_F=0$. On the other hand, in the subregion above the curve, the $\ell_1^*$ heuristic fails and one can even find an $X_{sol}$ for which $\textbf{RMSE}\rightarrow\infty$.

\begin{figure}[htb]
			\centerline{\includegraphics[width=.7\columnwidth]{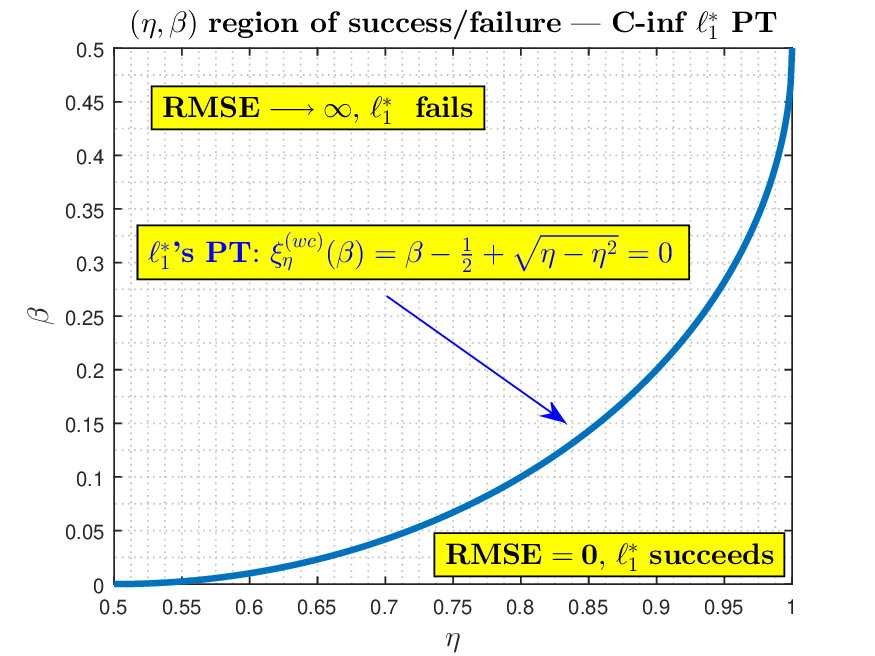}}
\caption{Typical \emph{\textbf{worst case}} $\ell_1^*$ phase transition (ideal low rank context (block causal inference  -- C-inf))}
\label{fig:cinfspectypwcPTbetaeta}
\end{figure}

As mentioned after (\ref{eq:cinfanl5a}), one should observe that (\ref{eq:cinfanl5a}) holds as long as the optimizing condition is feasible. The above theorem precisely identifies the relation between the system parameters $\beta$ and $\eta$ when that happens. Visually speaking, that happens whenever $(\eta,\beta)$ are below (or to the right of) the curve in Figure \ref{fig:cinfspectypwcPTbetaeta}. As soon as the system parameters are above the curve, the optimizing condition becomes infeasible, and one has that both the  optimization in (\ref{eq:cinfanl5})  and the residual error norm, $\|W\|_F$, are, in general, unbounded. Moreover, as mentioned immediately after (\ref{eq:cinfanl5a}), the feasibility of the optimizing condition is precisely what determines the phase transition in the idealized scenario (where $\epsilon_\sigma=0$) and what is fully characterized by Theorem \ref{thm:typwcthm1}.

In regimes below the phase transitions, the residual error norm, denoted as $\|W\|_F$, can be effectively characterized. We provide a detailed characterization of this in the subsequent section.

\subsection{Operating in the regimes below phase transitions}
\label{sec:belowpt}

As mentioned above, one obtains the worst case scenario for $V=U$. In such a scenario one also has $V_B=V_C$, $\Lambda$ becomes symmetric and its parametrization from (\ref{eq:cinfanl11aa1z1}) and (\ref{eq:cinfanl11aa2z1}) becomes
	\begin{equation}\label{eq:supp1}
		\Lambda = V_BEV_B^T+V_B F(V_B^{\perp})^T+V_B^\perp D.
	\end{equation}
Moreover, although it is not strictly needed for our considerations here, we mention that from (\ref{eq:cinfcor2eq1}), (\ref{eq:cinfanl11aa1z1}) and (\ref{eq:cinfanl11aa2z1}) one can explicitly determine
	\begin{equation}\label{eq:supp2}
E=-\Sigma_B^{-1}V_B^TGV_B\Sigma_B^{-1}, \quad \mbox{where}\quad G=(I^{(l)})^T \bar{U}^{\perp} (\bar{U}^{\perp})^T I^{(l)}.
	\end{equation}
Recalling (\ref{eq:cinfanl5a}) and utilizing Proposition \ref{cor:cinfcor3} together with the above parametrization of $\Lambda$ from (\ref{eq:supp1}), we further have
{\small  \begin{eqnarray}
f_{pr}(M;U,V) & = &   \max_{\Lambda,\Lambda^T\Lambda\leq I,(I^{(l)})^T(\bV\bU^T+\bV^{\perp}\Lambda(\bU^{\perp})^TI^{(l)}=0} \tr(\Lambda\diag(\epsilon_\sigma)) \nonumber \\
    & = &   \max_{\Lambda = V_BEV_B^T+V_B F(V_B^{\perp})^T+V_B^\perp D,\Lambda\Lambda\leq I} \tr((V_BEV_B^T+V_B F(V_B^{\perp})^T+V_B^\perp D)\diag(\epsilon_\sigma)). \nonumber \\\label{eq:supp3}
\end{eqnarray}}We continue to consider the worst case and take $\epsilon_\sigma=\max(\epsilon_\sigma)\1$ which implies $\diag(\epsilon_\sigma)=\max(\epsilon_\sigma) I$. Moreover, since it will turn out that, in the worst case, the particular scaling plays no role, to facilitate writing we will assume $\diag(\epsilon_\sigma)= I$. We then from (\ref{eq:supp3}) have
\begin{eqnarray}
f_{pr}(M;U,V) & = &   \max_{\Lambda = V_BEV_B^T+V_B^\perp D_S (V_B^\perp)^T,\Lambda\Lambda\leq I} \tr((V_BEV_B^T+V_B^\perp D_S (V_B^\perp)^T)\diag(\epsilon_\sigma)). \nonumber \\
& = &  \tr(V_BEV_B^T) +n-k-l,
\label{eq:supp4}
\end{eqnarray}
where the optimum is achieved for $D_S=I$ and ultimately
 	\begin{equation}\label{eq:supp5}
		\tilde{\Lambda} = V_BEV_B^T+V_B^\perp D_S (V_B^\perp)^T= V_BEV_B^T+V_B^\perp (V_B^\perp)^T.
	\end{equation}
We then have the following key result.
\begin{theorem}(Residual optimization characterization)
 Assume the setup of Theorem \ref{thm:thm1}. One then has the following
\vspace{-.0in}\begin{center}
\tcbset{beamer,lower separated=false, fonttitle=\bfseries,
coltext=black , interior style={top color=orange!10!yellow!30!white, bottom color=yellow!80!yellow!50!white}, title style={left color=orange!10!cyan!30!blue, right color=green!70!blue!20!black}}
 \begin{tcolorbox}[beamer,title=\hspace{-.12in}\textbf{\yellow{Worst case RMSE (deterministic) estimate}:},width=1.0\linewidth]
{\small \vspace{-.07in}\begin{equation}
\hspace{-.19in}\|\bar{W}\|_F \hspace{-.02in}=\hspace{-.02in}\sqrt{-\hspace{-.02in}\min_{\bar{\Lambda},\bar{\Lambda}(I-(((I^{(l)})^T\bV^{\perp})^T(I^{(l)})^T\bV^{\perp}))=0}
\hspace{-.02in} \lp \frac{1}{4}\|(I^{(l)})^T(\bV^{\perp}\bar{\Lambda}(\bV^{\perp})^TI^{(l)}\|_F^2  -   \tr(\bar{\Lambda}\diag(\epsilon_\sigma)) \rp}.\label{eq:cinfthm2eq1a0}
\end{equation}}
\vspace{-.17in}
 \end{tcolorbox}
\end{center}
\label{thm:thm2}
\end{theorem}
\begin{proof}
  We start by considering the following optimization
\begin{eqnarray}
 \|\hat{W}\|_F^2  =    \min_{W} & & \|W\|_F^2. \nonumber \\
\mbox{subject to} & &  \tr(\bU^TW\bV)+\ell_1^*(\diag(\epsilon_\sigma)+(\bU^{\perp})^TW\bV^{\perp})\leq f_{pr}(M;U,V)=\tr(\tilde{\Lambda}) \nonumber \\
 & & M\circ W=0.
\label{eq:supp6}
\end{eqnarray}
Following again the Lagrangian mechanism discussed above, we write
 \begin{eqnarray}
\mathcal{L}(W,\lambda,\Theta) \hspace{-.0in} & = & \hspace{-.0in} \|W\|_F^2 +\lambda \tr(\bU^TW\bV)+\lambda \ell_1^*(\diag(\epsilon_\sigma)+(\bU^{\perp})^TW\bV^{\perp})
\nonumber \\
& & +\Theta (M\circ W) -\lambda\tr(\tilde{\Lambda})\nonumber \\
& = & \hspace{-.0in}\max_{\Lambda,\Lambda^T\Lambda\leq \lambda^2 I} \Bigg.\Bigg( \|W\|_F^2 + \lambda \tr(\bU^TW\bV)+\tr(\Lambda(I+(\bU^{\perp})^TW\bV^{\perp})) \nonumber \\
& & +\Theta (M\circ W) -\lambda\tr(\tilde{\Lambda})\Bigg.\Bigg)\nonumber \\
& = & \hspace{-.0in}\max_{\Lambda,\Lambda^T\Lambda\leq \lambda^2 I} \Bigg(\Bigg. \|W\|_F^2 + \tr\lp (\lambda \bV\bU^T+\bV^{\perp}\Lambda(\bU^{\perp})^T
+\Theta\circ M)W\rp
 \nonumber \\
& &  + \tr(\Lambda-\lambda \tilde{\Lambda}) \Bigg.\Bigg).
\label{eq:supp7}
\end{eqnarray}
We then also have for the Lagrange dual function
\begin{eqnarray}
g(\lambda,\Theta)\hspace{-.0in} & = & \hspace{-.0in} \min_{W}\mathcal{L}(W,\lambda,\Theta) \hspace{-.1in}  \nonumber \\
& = & \hspace{-.0in}\min_{W}\max_{\Lambda,\Lambda^T\Lambda\leq \lambda^2 I} \Bigg(\Bigg. \|W\|_F^2 + \tr\lp (\lambda \bV\bU^T+\bV^{\perp}\Lambda(\bU^{\perp})^T+\Theta\circ M)W\rp \nonumber \\
& & + \tr(\Lambda-\lambda \tilde{\Lambda}) \Bigg.\Bigg).\nonumber \\
\label{eq:supp8}
\end{eqnarray}
Connecting (\ref{eq:supp6})-(\ref{eq:supp8}) and utilizing the Lagrange duality gives
\begin{eqnarray}
\|\hat{W}\|_F^2 \hspace{-.0in} & \geq &  \hspace{-.0in} \max_{\Lambda,\Lambda^T\Lambda\leq \lambda^2 I} \min_{W}
\Bigg(\Bigg.
\|W\|_F^2 + \tr\lp (\lambda \bV\bU^T+\bV^{\perp}\Lambda(\bU^{\perp})^T+\Theta\circ M)W\rp
  \nonumber \\
&  &+ \tr(\Lambda-\lambda \tilde{\Lambda}) \Bigg.\Bigg) \nonumber \\
& =  & \hspace{-.0in} \max_{\Lambda,\Lambda^T\Lambda\leq \lambda^2 I}  \lp -\frac{1}{4}\|(I^{(l)})^T (\lambda \bV\bU^T+\bV^{\perp}\Lambda(\bU^{\perp})^T)I^{(l)}\|_F^2 + \tr(\Lambda-\lambda \tilde{\Lambda}) \rp \nonumber \\
& =  & \hspace{-.0in} -\min_{\Lambda,\Lambda^T\Lambda\leq \lambda^2 I}  \lp \frac{1}{4}\|(I^{(l)})^T (\lambda \bV\bU^T+\bV^{\perp}\Lambda(\bU^{\perp})^T)I^{(l)}\|_F^2 - \tr(\Lambda-\lambda \tilde{\Lambda}) \rp. \nonumber \\
\label{eq:supp9}
\end{eqnarray}
After setting
\begin{eqnarray}
 \bar{\Lambda}=\Lambda-\lambda\tilde{\Lambda},
 \label{eq:supp10}
\end{eqnarray}
one can rewrite  (\ref{eq:supp9}) as
{\small \begin{eqnarray}
\|\hat{W}\|_F^2\hspace{-.0in}
& =  & \hspace{-.0in} -\min_{\Lambda,\Lambda^T\Lambda\leq \lambda^2 I}  \lp \frac{1}{4}\|(I^{(l)})^T (\lambda \bV\bU^T+\bV^{\perp}\Lambda(\bU^{\perp})^T)I^{(l)}\|_F^2 - \tr(\Lambda-\lambda \tilde{\Lambda}) \rp \nonumber \\
& =  & \hspace{-.0in} -\min_{\Lambda,(\bar{\Lambda}+\lambda\tilde{\Lambda})^T(\bar{\Lambda}+\lambda\tilde{\Lambda})\leq \lambda^2 I}  \lp \frac{1}{4}\|(I^{(l)})^T (\lambda \bV\bU^T+\bV^{\perp}(\bar{\Lambda}+\lambda \Lambda) (\bU^{\perp})^T)I^{(l)}\|_F^2 - \tr(\bar{\Lambda}) \rp \nonumber \\
& =  & \hspace{-.0in} -\min_{\Lambda,(\bar{\Lambda}+\lambda\tilde{\Lambda})^T(\bar{\Lambda}+\lambda\tilde{\Lambda})\leq \lambda^2 I}  \lp \frac{1}{4}\|(I^{(l)})^T \bV^{\perp}\bar{\Lambda}(\bU^{\perp})^TI^{(l)}\|_F^2 - \tr(\bar{\Lambda}) \rp. \nonumber \\
\label{eq:supp11}
\end{eqnarray}}As mentioned above, one has in the worst case $V=U$ and $\bar{V}=\bar{U}$ and the following parametrization of $\bar{\Lambda}$
\begin{eqnarray}
 \bar{\Lambda}=\begin{bmatrix}V_B^\perp & V_B  \end{bmatrix}\begin{bmatrix}Z_1 & Z_2 \\Z_2 & Z_3 \end{bmatrix}
 \begin{bmatrix}V_B^\perp & V_B \end{bmatrix}^T.
 \label{eq:supp12}
\end{eqnarray}
After recalling on parametrization of $\tilde{\Lambda}$ from (\ref{eq:supp5}) we find
\begin{eqnarray}
 \tilde{\Lambda}=\begin{bmatrix}V_B^\perp & V_B  \end{bmatrix}\begin{bmatrix}I & 0 \\0 & E \end{bmatrix}
 \begin{bmatrix}V_B^\perp & V_B \end{bmatrix}^T,
 \label{eq:supp13}
\end{eqnarray}
and
\begin{eqnarray}
 \bar{\Lambda}+ \lambda \tilde{\Lambda}=\begin{bmatrix}V_B^\perp & V_B  \end{bmatrix}\begin{bmatrix}\lambda I+Z_1 & Z_2 \\Z_2 & \lambda E+Z_3 \end{bmatrix}
 \begin{bmatrix}V_B^\perp & V_B \end{bmatrix}^T.
 \label{eq:supp14}
\end{eqnarray}
Moreover, we then also have
\begin{eqnarray}
\lambda^2 I- (\bar{\Lambda}+ \tilde{\Lambda})^2=\begin{bmatrix}V_B^\perp & V_B  \end{bmatrix}\begin{bmatrix}\lambda^2 I-(\lambda I+Z_1)^2-Z_2^TZ_2 & Z_4 \\ Z_4 & Z_5 \end{bmatrix}
 \begin{bmatrix}V_B^\perp & V_B \end{bmatrix}^T.
 \label{eq:supp15}
\end{eqnarray}
For $\lambda^2 I- (\bar{\Lambda}+ \tilde{\Lambda})^2$ to be positive semi-definite one would need $\lambda^2 I-(\lambda I+Z_1)^2-Z_2^TZ_2\geq 0$ which means that $Z_1\leq 0$. Recalling from (\ref{eq:cinfanl11aa1z1}), $(I^{(l)})^TV^\perp=U_B\Sigma_BV_B^T$, one then from (\ref{eq:supp11}) has
that $\bar{\Lambda}$ does not belong to the span of $V_B^\perp$, i.e. it does not belong to the span orthogonal to $(I^{(l)})^TV^\perp$. Together with (\ref{eq:supp11}) and the strong duality (which, due to the underlying convexity, ensures that the inequality in (\ref{eq:supp9}) can be replaced by an equality) this gives the equation in the statement of Theorem \ref{thm:thm2} and ultimately completes the prof of the theorem.
\end{proof}

A combination of Theorem \ref{thm:thm2} and Lemma \ref{lemma:lemma1} gives the following statistical  \textbf{RMSE} characterization.

\begin{theorem}(Worst case \textbf{RMSE})
 Assume the setup of Theorem \ref{thm:thm1} with $U=V$ Haar distributed and a large $n$ linear regime with $\beta\triangleq \lim_{n\rightarrow \infty}\frac{k}{n}$ and $\eta\triangleq \lim_{n\rightarrow \infty}\frac{l}{n}$. Assume that $\hat{X}$ is the solution of (\ref{eq:genmcl1posmmt}). One then has the following
\vspace{-.0in}\begin{center}
\tcbset{beamer,lower separated=false, fonttitle=\bfseries,
coltext=black , interior style={top color=orange!10!yellow!30!white, bottom color=yellow!80!yellow!50!white}, title style={left color=orange!10!cyan!30!blue, right color=green!70!blue!20!black}}
 \begin{tcolorbox}[beamer,title=\hspace{-.12in}\textbf{\yellow{Worst case RMSE (statistical) estimate:}},width=5.2in]
\vspace{-.07in}\begin{eqnarray}
\hspace{-.2in}   x_{l/u} & \triangleq & \beta+\eta-2\beta\eta  \pm  \sqrt{(\beta+\eta-2\beta\eta)^2-(\beta-\eta)^2}, \nonumber \\
\hspace{-.2in}\xi &  \triangleq &
              \sigma_\epsilon \frac{\sqrt{\int_{x_l}^{x_u}\frac{\sqrt{-(x-(\beta+\eta))^2-4\beta\eta(x-1)}}{2\pi x^2(x-x^2)}dx+\max(1-(\beta+\eta),0)}}{\sqrt{1-\beta}(1-\eta)}, \nonumber \\
\hspace{-.2in}\forall X_{sol} & &  \lim_{n\rightarrow\infty}  \mP\lp\textbf{RMSE}  =  \frac{n\|\hat{X}-X_{sol}\|_F}{\sqrt{n-k}(n-l)}  \leq  \xi\rp \longrightarrow 1, \nonumber \\
\hspace{-.2in}\exists X_{sol} & &  \lim_{n\rightarrow\infty}  \mP\lp \textbf{RMSE}  =  \frac{n\|\hat{X}-X_{sol}\|_F}{\sqrt{n-k}(n-l)}  \in ((1-\epsilon)\xi,(1+\epsilon)\xi)\rp \longrightarrow 1.\hspace{.2in}\nonumber \\
 \label{eq:cinfthm3eq1a0}
\end{eqnarray}
\vspace{-.17in}
 \end{tcolorbox}
\end{center}
\label{thm:thm3}
\end{theorem}
\begin{proof}
From considerations presented in (\ref{eq:supp11})-(\ref{eq:supp15}) we have
\begin{equation}
\|\hat{W}\|_F^2
 =   -\min_{Z_3,\bar{\Lambda}=V_BZ_3V_B^T,\lambda, (\bar{\Lambda}+\lambda\tilde{\Lambda})^T(\bar{\Lambda}+\lambda\tilde{\Lambda})\leq \lambda^2 I}  \lp \frac{1}{4}\|(I^{(l)})^T \bU^{\perp}\bar{\Lambda}(\bU^{\perp})^TI^{(l)}\|_F^2 - \tr(\bar{\Lambda}) \rp .
\label{eq:supp16}
\end{equation}
Substituting  $(I^{(l)})^T \bU^{\perp}=U_B\Sigma_BB_B^T$ we obtain
\begin{equation}
\|\hat{W}\|_F^2
 =   -\min_{Z_3,\lambda, (Z_3+\lambda E)^T(Z_3+\lambda E)\leq \lambda^2 I}  \lp \frac{1}{4}\tr(\Sigma_BZ_3\Sigma_B^2Z_3\Sigma_B) - \tr(Z_3) \rp .
\label{eq:supp17}
\end{equation}
 Elementary algebraic considerations give that optimal $Z_3$ is diagonal and since $\lambda$ can be arbitrarily large we have from (\ref{eq:supp17})
\begin{equation}
\|\hat{W}\|_F^2
 =   -\min_{\bar{z}}  \lp \frac{1}{4}\sum_{i=1}^{n-l}\bar{z}_i^2\bar{\sigma}_i^4 - \sum_{i=1}^{n-l}\bar{z}_i \rp=\sum_{i=1}^{n-l}\frac{1}{\bar{\sigma}_i^4}, \quad \mbox{where}\quad \bar{\sigma}=\diag(\Sigma_B).
\label{eq:supp18}
\end{equation}
Keeping in mind that the distribution of $\bar{\sigma}_i^2$ is as given in  (\ref{eq:typwclemma2eq4}) and (\ref{eq:typwclemma2eq5}), one finally has
\begin{eqnarray}
\mE\|\hat{W}\|_F^2
 & = &  \mE \sum_{i=1}^{n-l}\frac{1}{\bar{\sigma}_i^4} \nonumber \\
 & = &\lp \int_{x_l}^{x_u}
              \frac{\sqrt{-(x-(\beta+\eta))^2-4\beta\eta(x-1)}}{2\pi x^2(x-x^2)} + \max(1-(\beta+\eta),0)\delta(x-1)\rp.\nonumber \\
\label{eq:supp19}
\end{eqnarray}
Concentrations, scalings by $\sqrt{n-k}(n-l)/n$, and adjusting the maximum magnitude of $\epsilon_{\sigma}$, $\max(\epsilon_{\sigma})$, to be any quantity $\sigma_{\epsilon}$ (instead of 1) completes the proof of the theorem.
\end{proof}

\vspace{-.0in}
\subsection{A summarized interpretation of the methodology and obtained theoretical results}
\label{sec:interpret}

We find it useful to provide a more concrete interpretation of the results given in Theorem \ref{thm:thm3} within the causal inference context (see, e.g., \cite{ABDIK21} for a pioneering work in establishing matrix completion $\longleftrightarrow$ block causal inference (MC$\longleftrightarrow$ block C-inf) connection). We first recall that within the MC $\longleftrightarrow$ block C-inf context, the rows and columns of $X_{sol}$ correspond to the units and time periods (at which the units can potentially be exposed to treatments), respectively. Assuming the so-called irreversible simultaneous treatments during $n-l$ consecutive time periods on a subset of $n-l$ units, the block missing part of $X_{sol}$ corresponds to the so-called \emph{counterfactuals} (values of the units had the treatment not been applied). The approximate recovery of the counterfactual part of $X_{sol}$ is then the ultimate block C-inf's goal. The optimization problem in (\ref{eq:genmcl0posmmt}) is posed so that its optimal solution hopefully closely matches $X_{sol}$.  However, solving (\ref{eq:genmcl0posmmt}) in an idealized scenario is NP hard and in an non-idealized one even uniquely impossible. As is a common practice in matrix completion literature, we consider $\ell_1^*$ heuristic as the tightest convex relaxation of $\ell_0^*$ (which, in the idealized scenario, is known to often be successful) hoping that its solution, $\hat{X}$, also closely matches $X_{sol}$. Under statistical assumptions on the origin of $X_{sol}$, Theorem \ref{thm:thm3} effectively provides the exact relation among three key system and algorithmic parameters: (\textbf{\emph{i}}) data heterogeneity ($\text{rank}(X_{sol})$ or in a scaled form $\beta=\lim_{n\rightarrow\infty}\frac{\text{rank}(X_{sol})}{n}=\lim_{n\rightarrow\infty}\frac{k}{n}$); (\textbf{\emph{ii}}) the number of the untreated units and the treatment starting time ($l$ or in a scaled form $\eta=\lim_{n\rightarrow\infty}\frac{l}{n}$); (\textbf{\emph{iii}}) the absolute spectral deviation from the ideal low rankness ($\sigma_{\epsilon}$); and (\textbf{\emph{iv}}) the residual $\textbf{RMSE}$ error, $\|\hat{X}-X_{sol}\|_F$.

As particularly relevant, we also emphasize the following two observed features of the obtained theoretical results: (\textbf{\emph{i}}) (\textbf{\emph{Simplicity}}) Both the phase transition characterization from Theorem \ref{thm:typwcthm1} and the corresponding residual $\textbf{RMSE}$ one from Theorem \ref{thm:thm3} are remarkably simple, single line functional equations where the key system parameters $(\beta,\eta,\sigma_{\epsilon})$ explicitly appear. (\textbf{\emph{ii}}) (\textbf{\emph{Generality}}) The characterization from Theorem \ref{thm:thm3} exhibits a particularly generic behavior. Namely, the residual error is characterized in the worst case sense, ensuring that no matter what the relation between $\sigma$ and $\sigma_{\epsilon}$ is, the error remains bounded by the quantity given in the theorem. In other words, no matter what the level of relative deviation from the ideal low rankness is, the error is always bounded by the quantity given in the theorem. Moreover, for $\sigma \gg \sigma_{\epsilon}$, the bound is achievable, which implies that better (more generic) results cannot be proven.


\vspace{ -.0in}

\section{Numerical results}
\label{sec:cpmcnumres}

\vspace{ -.0in}

To complement the above theoretical findings, we conduct a set of numerical experiments in several different relevant regimes. 
Our focus is on comparing our theoretical $\ell_1^*$ performance analysis results with the corresponding $\ell_1^*$ numerical estimates.

\textbf{\emph{1) Worst case regime:}} As the theory suggests, somewhat counter-intuitively, the worst case $\textbf{RMSE}$ estimates are achieved for matrices with a large gap between dominant singular values and the negligible ones. To emulate such a regime we set $\sigma_\epsilon=1$ and magnitude of each of the large $\sigma$'s components to be 50. We find the ratio 50:1 as sufficiently large to emulate the theoretically needed infinite one. The obtained results are shown in Table \ref{tab:tab2}. It is easily observed that there is a close agreement between theoretical and simulated results across the range of allowed $\beta$. One additionally notes that for $\beta=0.2$ a slight deviation from what the theory predicts can be seen. The reason, of course, is the fact that at $\beta=0.2$, one is exactly on the phase transition and above such a $\beta$ the residual error can even be infinite.

\begin{table}[h]
  \caption{Block missingness matrix completion via nuclear norm minimization -- \emph{worst case} $\textbf{RMSE}$; $\eta=0.9$}
  \label{tab:tab2}
  \centering
  \begin{tabular}{ccccc}
    \hline
  \textbf{Method}  & \multicolumn{4}{c}{$\beta$}                   \\
    \cline{2-5}
    (to estimate $\textbf{RMSE}$)     & $\mathbf{0.05}$     & $\mathbf{0.1}$  & $\mathbf{0.15}$ & $\mathbf{0.2}$ \\
    \hline
 $ $ \hspace{.2in} \bl{\textbf{theoretical} $\frac{n\mE\|\hat{X}-X_{sol}\|_F}{\sigma_\epsilon\sqrt{n-k}(n-l)}$}    \hspace{.2in} $ $  & \bl{$\mathbf{3.446}$}  & \bl{$\mathbf{3.755}$} & \bl{$\mathbf{4.161}$}   & \bl{$\mathbf{4.617}$}    \\
    \hline
$ $ \hspace{.2in}   \prp{\textbf{simulated} $\frac{n\mE\|\hat{X}-X_{sol}\|_F}{\sigma_\epsilon\sqrt{n-k}(n-l)}$}  \hspace{.2in} $ $ & \prp{$\mathbf{3.46}$} & \prp{$\mathbf{3.82}$}  & \prp{$\mathbf{4.24}$} & \prp{$\mathbf{5.00}$}     \\
    \hline
$ $ \hspace{.2in}   \green{\textbf{asymmetric: \hspace{.1in} simulated} $\frac{n\mE\|\hat{X}-X_{sol}\|_F}{\sigma_\epsilon\sqrt{n-k}(n-l)}$}  \hspace{.2in} $ $ & \green{$\mathbf{2.33}$} & \green{$\mathbf{2.66}$}  & \green{$\mathbf{3.18}$} & \green{$\mathbf{3.73}$}     \\
      \hline
  \end{tabular}
\end{table}

\textbf{\emph{2) Asymmetric case:}} Differently from $U=V$ which achieves the worst case but is highly unlikely to be observed in real data, we consider its an asymmetric anti-pod, where not only aren't $U$ and $V$ equal to each other, they are actually completely independent of each other. This scenario is expected to more faithfully capture real data. In the absence of any other \emph{a priori} available knowledge about the underlying parts of the model, such a scenario might in fact be the most natural one to assume. We show the results for this scenario with identical system parameters as in the worst case. As the results in Table \ref{tab:tab2} show, the worst case results are indeed upper bounds on the ones obtained in the asymmetric scenario.

\textbf{\emph{3) Varying $\epsilon_\sigma$:}} For the worst case scenario, we set the magnitude of $\epsilon_\sigma$ equal to $\sigma_\epsilon=1$. We also simulated the scenario where the magnitudes are standard normals and uniformly random while we maintained the overall norm of $\epsilon_\sigma=\sqrt{n-k}$ to ensure fairness in comparison with the worst case scenario. The results are shown in Table \ref{tab:tab3}. As it can be seen from the table, the results are dominated by the worst case ones as expected by the theoretical predictions. One should though observe that the results obtained under the Gaussian assumption are actually fairly close to the flat ones (where all the magnitudes of $\epsilon_\sigma$ are equal to $\sigma_\epsilon=1$).

\begin{table}[h]
  \caption{Block missingness matrix completion via nuclear norm minimization -- \emph{worst case} $\textbf{RMSE}$; $\eta=0.9$; different statistics of $\epsilon_\sigma$}
  \label{tab:tab3}
  \centering
  \begin{tabular}{ccccc}
    \hline
  \textbf{Method}  & \multicolumn{4}{c}{$\beta$}                   \\
    \cline{2-5}
    (to estimate $\textbf{RMSE}$)     & $\mathbf{0.05}$     & $\mathbf{0.1}$  & $\mathbf{0.15}$ & $\mathbf{0.2}$ \\
    \hline
 $ $ \hspace{.2in} \bl{\textbf{theoretical} $\frac{n\mE\|\hat{X}-X_{sol}\|_F}{\sigma_\epsilon\sqrt{n-k}(n-l)}$}    \hspace{.2in} $ $  & \bl{$\mathbf{3.446}$}  & \bl{$\mathbf{3.755}$} & \bl{$\mathbf{4.161}$}   & \bl{$\mathbf{4.617}$}    \\
    \hline
$ $ \hspace{.2in}   \prp{\textbf{Flat $\epsilon_\sigma$: \hspace{.1in} simulated} $\frac{n\mE\|\hat{X}-X_{sol}\|_F}{\sigma_\epsilon\sqrt{n-k}(n-l)}$}  \hspace{.2in} $ $ & \prp{$\mathbf{3.46}$} & \prp{$\mathbf{3.82}$}  & \prp{$\mathbf{4.24}$} & \prp{$\mathbf{5.00}$}     \\
    \hline
$ $ \hspace{.2in}   \green{\textbf{Gaussian $\epsilon_\sigma$: \hspace{.1in} simulated} $\frac{n\mE\|\hat{X}-X_{sol}\|_F}{\sigma_\epsilon\sqrt{n-k}(n-l)}$}  \hspace{.2in} $ $ & \green{$\mathbf{3.40}$} & \green{$\mathbf{3.54}$}  & \green{$\mathbf{3.92}$} & \green{$\mathbf{5.00}$}     \\
    \hline
$ $ \hspace{.2in}   \green{\textbf{Uniform $\epsilon_\sigma$: \hspace{.1in} simulated} $\frac{n\mE\|\hat{X}-X_{sol}\|_F}{\sigma_\epsilon\sqrt{n-k}(n-l)}$}  \hspace{.2in} $ $ & \green{$\mathbf{3.15}$} & \green{$\mathbf{3.44}$}  & \green{$\mathbf{3.83}$} & \green{$\mathbf{5.03}$}     \\
      \hline
  \end{tabular}
\end{table}

\textbf{\emph{4) Varying magnitudes of $\sigma$:}} As discussed above, to achieve the worst case behavior, one needs to choose a rather large magnitude of dominating singular values $\sigma$. A natural question is then, what happens when the magnitudes of $\sigma$, $\sigma_{mag}$, are not necessarily as large. We, in Table \ref{tab:tab4}, show the obtained results as magnitudes of $\sigma$ change (to ensure a fair comparison, we continue to keep, as earlier,  $\epsilon_\sigma$ equal to $\sigma_\epsilon=1$). As it can be seen from the table, the results are converging in a rather fast manner. The ratios $\frac{\sigma_{mag}}{\sigma_\epsilon}=\frac{50}{1}$ that we, out of precaution, chose above is not even needed. The worst case behavior is fairly closely approached, already for ratios around $6:1$.

\begin{table}[h]
  \caption{Block missingness matrix completion via nuclear norm minimization; $\eta=0.9$; $\textbf{RMSE}$ as a function of $\frac{\sigma_{mag}}{\sigma_\epsilon}$}
  \label{tab:tab4}
  \centering
  \begin{tabular}{ccccccc}
    \hline
  \textbf{Method}  & \multicolumn{5}{c}{$\frac{\sigma_{mag}}{\sigma_\epsilon}$}  & \bl{$\frac{\sigma_{mag}}{\sigma_\epsilon}$ (theory)}                 \\
    \cline{2-6}
    (to estimate $\textbf{RMSE}$)     & $\mathbf{1}$  & $\mathbf{2}$ & $\mathbf{3} $  & $\mathbf{4}$     & $\mathbf{6}$  &  \bl{$\mathbf{\infty}$} \\
     \hline
   \prp{\textbf{Flat $\epsilon_\sigma$: \hspace{.1in} simulated} $\frac{n\mE\|\hat{X}-X_{sol}\|_F}{\sigma_\epsilon\sqrt{n-k}(n-l)}$}  \hspace{.1in} $ $
& \prp{$\mathbf{3.43}$} & \prp{$\mathbf{3.71}$}  & \prp{$\mathbf{3.87}$} & \prp{$\mathbf{3.93}$} & \prp{$\mathbf{4.01}$} & \bl{$\mathbf{4.16}$}    \\
       \hline
  \end{tabular}
\end{table}

We choose $n=40$, as a rather small matrix dimension. Despite the fact that the  proofs are obtained assuming large dimensional settings, we observe a strong agreement between theoretical and simulated results even for fairly small sizes. Also, to show an agreement with the theoretical results, we, consider the so-called \emph{typical} behavior with the singular values of the unknown targeted  matrices equal to one. Choices of this type are known to serve as the worst case examples in establishing the reversal $\ell_0^*-\ell_1^*$-equivalence conditions in idealized scenarios (for completeness, we also ran simulations where the singular values were randomly chosen and the obtained results were either identically matching or improving upon the ones shown in Table \ref{tab:tab2}).

\vspace{ -.0in}

\section{Conclusion}
\label{sec:conc}

\vspace{ -.1in}

In this paper, we studied the nuclear norm minimization heuristic and theoretically analyzed its performance for completing matrices with \emph{missing not at random} (MNAR) entries. We place special emphasis on \emph{approximately low rank} matrices  and on the \emph{block missing patterns} encountered in the so-called irreversible simultaneous treatment scenarios often seen in various health and social science applications. In such scenarios, an exact recovery of the incomplete matrices is not attainable and the residual errors are present as well. To assess the quality of the underlying nuclear norm based algorithms, we analytically characterized the residual errors.

Relying on a typical statistical view of the problem and utilizing random matrix theory  related to free probability theory (FPT), we first introduced a generic performance analysis framework. We then conducted a concrete, statistical phase transition type of analysis, and obtained \emph{explicit exact worst case} estimates for the residual $\textbf{RMSE}$ which hold below the so-called (worst case) phase-transition line. The resulting characterizations are simple, single line functional equations where all three key system parameters --- (\textbf{\emph{i}}) the data heterogeneity ($\beta$), (\textbf{\emph{ii}}) the size of the missing blocks ($\eta$), and (\textbf{\emph{iii}}) the magnitude of the spectral deviations from the ideal low rankness ($\sigma_{\epsilon}$)--- are clearly visible.

To complement our theoretical findings, we conducted several numerical experiments as well. The following are the main takeaways: \textbf{\emph{(i)}} A very strong agreement between the theoretical predictions and numerical simulations is observed. \textbf{\emph{(ii)}} Despite the fact that our theoretical predictions are obtained in a large dimensional asymptotic regime, we observe a strong match between theory and simulations, even for rather small problems with dimensions on the order of few tens.

\newpage




%

\appendix
\paragraph{Appendix}
	This appendix contains proofs of technical results stated in the main body of the paper. We also demonstrate that our analysis can handle data datasets encountered in real applications, such as the  California tobacco data, widely used in many studies (e.g. \cite{ADHsynth10}). Our analysis, particular examining the spectrum of this data set, reveals that its approximately low rank structure is directly attributable to the data's inherent structure.

\section{Proof of Theorem~\ref{thm:thm1}}

The theorem follows from the proofs of  much stronger fundamental propositions and corollaries that relate to a more generic concept of the low rank recovery (LRR). In such a concept, the fairly particular linear constraints of ~\eqref{eq:genmcl0posmmt} and~\eqref{eq:genmcl1posmmt} are replaced by a more generic $\y=\vecw(Y)=A\vecw(X)$, where $A\in\mR^{m\times n^2}$ is a generic measurement matrix.

\begin{proposition}\label{prop:cauinf}(\textbf{\bl{$\ell_0^*-\ell_1^*$-equivalence condition (LRR)}} -- \textbf{symmetric}  $X$)
Consider a $\bU\in\mR^{n\times k}$ such that $\bU^T\bU=I_{k\times k}$ and a $\rankw-k$ \textbf{a priori known to be} symmetric matrix $X_{sol}=X\in\mR^{n\times n}$  with all of its columns belonging to the span of $\bU$. For concreteness, and without loss of generality, assume that $X$ has only positive nonzero eigenvalues. For a given matrix $A\in\mR^{m\times n^2}$ ($m\leq n^2$) assume that $\y=A\vecw(X)\in \mR^m$. If
\begin{equation}
(\forall W\in \mR^{n\times n} | A\vecw(W)=\0_{m\times 1},W=W^T\neq \0_{n\times n}) \quad  -\tr(\bU^TW\bU)< \ell_1^*((\bU^{\perp})^TW\bU^{\perp}),
\label{eq:genmcposthmcond1}
\end{equation}
then the solutions of~\eqref{eq:genmcl0posmmt} and~\eqref{eq:genmcl1posmmt} coincide. Moreover, if
\begin{equation}
(\exists  W\in \mR^{n\times n} | A\vecw(W)=\0_{m\times 1},W=W^T\neq \0_{n\times n}) \quad  -\tr(\bU^TW\bU)\geq \ell_1^*((\bU^{\perp})^TW\bU^{\perp}),
\label{eq:genmcposthmcond2}
\end{equation}
then there is an $X$ from the above set of the symmetric matrices with columns belonging to the span of $\bU$  such that the solutions of~\eqref{eq:genmcl0posmmt}) and~\eqref{eq:genmcl1posmmt} are different.
\label{thm:genmcthmregposcond}
\end{proposition}

\begin{proof}
The proof can be obtained by a step-by-step adaptation of Theorem 2 in \cite{StojnicICASSP09} to matrices or their singular/eigenvalues. 
For concreteness and without loss of generality we also assume that the eigen-decomposition of $X$ is

\begin{eqnarray}\label{eq:genmcrec1}
X=U\Lambda U^T=\begin{bmatrix}
	\bU & \bU^{\perp}
\end{bmatrix}
\begin{bmatrix}
	\bar{\Lambda}_X & \0_{k\times (n-k)} \\
	\0_{(n-k)\times k} & \bar{\Lambda}_X^{\perp}
\end{bmatrix}
\begin{bmatrix}
	\bU & \bU^{\perp}
\end{bmatrix}^T.
\end{eqnarray}

\bl{ \underline{\textbf{(i) $\Longrightarrow$ (the if part):}}} Following step-by-step the proof of Theorem $2$ in \cite{StojnicICASSP09}, we start by  assuming that  $\hat{X}$ is the solution of~\eqref{eq:genmcl1posmmt}. Then we want to show that if~ \eqref{eq:genmcposthmcond1} holds then $\hat{X}=X$. As usual, we instead of that, assume opposite, i.e. we assume that (\ref{eq:genmcposthmcond1}) holds but $\hat{X}\neq X$. Then since $\y=A\vecw(\hat{X})$ and $\y=A\vecw(X)$ must hold simultaneously there must exist $W$ such that $\hat{X} =X+W$ with $W\neq 0$, $A\vecw(W)=0$. Moreover, since $\hat{X}$ is the solution of~\eqref{eq:genmcl1posmmt} one must also have
\begin{eqnarray}
\begin{array}{r r r l@{\ }}
	& \ell_1^*(X+W) = \ell_1^*(\hat{X}) & \leq & \ell_1^*(X) \\
	\Longleftrightarrow \hspace{.3in} $ $ & \ell_1^*(\begin{bmatrix}
		\bU & \bU^{\perp}
	\end{bmatrix}^T(X+W)\begin{bmatrix}
		\bU & \bU^{\perp}
	\end{bmatrix}) & \leq & \ell_1^*(X) \\
	\Longrightarrow \hspace{.3in} $ $ &\ell_1^*(\bU^T (X+W)\bU)+\ell_1^*((\bU^{\perp})^T (X+W)\bU^{\perp}) & \leq & \ell_1^*(X).
\end{array}\nonumber \\
\label{eq:genmcproof1}
\end{eqnarray}
The last implication follows after one trivially notes
\begin{eqnarray}
\ell_1^*(\begin{bmatrix}
	\bU & \bU^{\perp}
\end{bmatrix}^T(X+W)\begin{bmatrix}
	\bU & \bU^{\perp}
\end{bmatrix}) & = & \max_{\Lambda_*=\Lambda_*^T\in {\cal L}_*}
\tr(\Lambda_*
\begin{bmatrix}
	\bU & \bU^{\perp}
\end{bmatrix}^T(X+W)\begin{bmatrix}
	\bU & \bU^{\perp}
\end{bmatrix})\nonumber \\
& \geq & \max_{\Lambda_*=\Lambda_*^T\in {\cal L}_{*}^0}
\tr(\Lambda_*
\begin{bmatrix}
	\bU & \bU^{\perp}
\end{bmatrix}^T(X+W)\begin{bmatrix}
	\bU & \bU^{\perp}
\end{bmatrix})\nonumber \\
& = & \ell_1^*(\bU^T (X+W)\bU)+\ell_1^*((\bU^{\perp})^T (X+W)\bU^{\perp}),\label{eq:genmcproof1a}
\end{eqnarray}
where
\begin{eqnarray}
{\cal L}_{*}^0 &\triangleq& \left \{\Lambda_*\in\mR^{n\times n} | \Lambda_*=\Lambda_*^T,\Lambda_*\Lambda_*^T\leq I, \Lambda_*=\begin{bmatrix}
	\Lambda_{*,1} & 0_{k\times (n-k)} \\
	0_{(n-k)\times k} & \Lambda_{*,2}
\end{bmatrix} \right \} \nonumber \\
&\subseteq& \left \{\Lambda_*\in\mR^{n\times n} | \Lambda_*=\Lambda_*^T,\Lambda_*\Lambda_*^T\leq I\right \} \triangleq   {\cal L}_{*}.\label{eq:genmcproof1b}
\end{eqnarray}

\tcbset{beamer,lower separated=false, fonttitle=\bfseries,width=3.4in, coltext=white,
colback=yellow!70!orange!40!white,title style={left color=cyan!40!black!80!purple, right color=red!60!yellow!40!orange!80!white},
width=(\linewidth-4pt)/4, equal height group=AT,before=,after=\hfill,fonttitle=\bfseries}
\tcbset{colback=red!25!white!70!green!15!yellow,colframe=red!95!white,width=(\linewidth-4pt)/4,
equal height group=AT,before=,after=\hfill,fonttitle=\bfseries,
interior style={left color=cyan!40!black!80!purple, right color=red!60!yellow!40!orange!80!white}}


Next, the key observation is to note that the absolute values can be removed in the nonzero part and that the $\ell_1^*(\cdot)$ can be ``\emph{replaced}" by $\tr(\cdot)$.  Such a simple observation is the most fundamental reason for all the success of the \bl{\textbf{RDT}} when used for the \textbf{exact} performance characterization of the structured objects' recovery. From (\ref{eq:genmcproof1}) we then have
\begin{eqnarray}
\begin{array}{r r r l@{\ }}
	& \ell_1^*(\bU^T (X+W)\bU)+\ell_1^*((\bU^{\perp})^T (X+W)\bU^{\perp}) & \leq & \ell_1^*(X)\\
	\Longrightarrow   \hspace{.3in} $ $ & \tr(\bU^T (X+W)\bU)+\ell_1^*((\bU^{\perp})^T (W)\bU^{\perp}) & \leq & \ell_1^*(X)\\
	\Longleftrightarrow   \hspace{.3in} $ $ & \tr(\bU^T W\bU)+\ell_1^*((\bU^{\perp})^T W\bU^{\perp}) & \leq & 0.
\end{array}\label{eq:genmcproof2}
\end{eqnarray}
We have arrived at a contradiction as the last inequality in (\ref{eq:genmcproof2}) is exactly the opposite of (\ref{eq:genmcposthmcond1}). This implies that our initial assumption $\hat{X}\neq X$ cannot hold and we therefore must have $\hat{X}=X$. This is precisely the claim of the first part of the theorem.

\bl{ \underline{ \textbf{ (ii) $\Longleftarrow$ (the only if part):}}}  We now assume that (\ref{eq:genmcposthmcond2}) holds, i.e.
\begin{equation}
(\exists  W\in \mR^{n\times n} | A\vecw(W)=\0_{m\times 1},W\neq \0_{n\times n}) \quad  -\tr((\bU)^TW\bU)\geq \ell_1^*((\bU^{\perp})^TW\bU^{\perp})\label{eq:genmcproof3}
\end{equation}
and would like to show that for such a $W$ there is a symmetric rank-$k$ matrix $X$ with the columns belonging to the span of $\bU$ such that $\y=A\vecw(X)$,  and the following holds
\begin{equation}
\ell_1^*(X+W)<\ell_1^*(X).\label{eq:genmcproof4}
\end{equation}

Existence of such an $X$ would ensure that it both, satisfies all the constraints in (\ref{eq:genmcl1posmmt})  and is not the solution of (\ref{eq:genmcl1posmmt}) . One can reverse all the above steps from (\ref{eq:genmcproof3}) to (\ref{eq:genmcproof1}) with strict inequalities and arrive at the first inequality in (\ref{eq:genmcproof1}) which is exactly (\ref{eq:genmcproof4}). There are two implications that cause problems in such a reversal process, the one in (\ref{eq:genmcproof3}) and the one in (\ref{eq:genmcproof1}). If these implications were equivalences everything would be fine. We address these two implications separately.

\underline{1) \textbf{the implication in (\ref{eq:genmcproof2}) -- particular $X$ to ``overwhelm" $W$:}} Assume $X=\bU\Lambda_x\bU^T$ with $\Lambda_x>0$ being a diagonal matrix with arbitrarily large elements on the main diagonal (here it is sufficient even to choose diagonal of $\Lambda_x$ so that its smallest element is larger than the maximum eigenvalue of $\bU^TW\bU$). Now one, of course, sees the main idea behind the ``removing the absolute values" concept from \cite{StojnicICASSP09}. Namely, for such an $X$ one has that $\ell_1^*(\bU^TX+W)\bU)=tr(\ell_1^*(\bU^TX+W)\bU))$ since for symmetric matrices  the $\ell_1^*(\cdot)$ (as the sum of the argument's \emph{absolute} eigenvalues) and $\tr(\cdot)$ (as the sum of the argument's eigenvalues) are equal. That basically means that when going backwards the second inequality in (\ref{eq:genmcproof2}) not only follows from the first one but also implies it as well. In other words, for $X=\bU\Lambda_x\bU^T$ (with $\Lambda_x>0$ and arbitrarily large)
\begin{eqnarray}
\begin{array}{r r r l@{\ }}
	& \tr(\bU^T W\bU)+\ell_1^*((\bU^{\perp})^T W\bU^{\perp}) & \leq & 0 \\
	\Longleftrightarrow  \hspace{.3in}  $ $ & \tr(\bU^T (X+W)\bU^)+\ell_1^*((\bU^{\perp})^T (W)\bU^{\perp}) & \leq & \ell_1^*(X)\\
	\Longleftrightarrow \hspace{.3in} $ $ & \ell_1^*(\bU^T (X+W)\bU)+\ell_1^*((\bU^{\perp})^T (X+W)\bU^{\perp}) & \leq & \ell_1^*(X),
\end{array}\label{eq:genmcproof5}
\end{eqnarray}
which basically mans that there is an $X$ that can ``overwhelm" $W$ (in the span of $\bU$) and ensures that the \bl{``\textbf{\emph{removing the absolute values}}"} is not only a \textbf{\emph{sufficient}} but also a \textbf{\emph{necessary}} concept for creating the relaxation  equivalence condition.

\underline{2) \textbf{the implication in (\ref{eq:genmcproof1}):}} One would now need to somehow show that the third inequality in (\ref{eq:genmcproof1}) not only follows from the second one but also implies it as well. This boils down to showing that inequality in (\ref{eq:genmcproof1a}) can be replaced with an equality or, alternatively, that ${\cal L}^0$ and ${\cal L}$ are provisionally equivalent. Neither of these statements is generically true. However, since we have a set of $X$ at our disposal there might be an $X$ for which they actually hold. We continue to assume $X=\bU\Lambda_x\bU^T$ with $\Lambda_x>0$ being a diagonal matrix with arbitrarily large entries on the main diagonal. Then the last equality in (\ref{eq:genmcproof1a}) gives
\begin{eqnarray}
\begin{array}{r r r l@{\ }}
	$ $ & \ell_1^*(\bU^T (X+W)\bU)+\ell_1^*((\bU^{\perp})^T (X+W)\bU^{\perp}) & \leq & \ell_1^*(X) \\
	\Longleftrightarrow \hspace{.3in} $ $ & \max_{\Lambda_*=\Lambda_*^T\in {\cal L}_{*}^0}
	\tr(\Lambda_*
	\begin{bmatrix}
		\bU & \bU^{\perp}
	\end{bmatrix}^T(X+W)\begin{bmatrix}
		\bU & \bU^{\perp}
	\end{bmatrix}) & \leq & \ell_1^*(X).
\end{array}\label{eq:genmcproof6}
\end{eqnarray}
Also, one has
\begin{eqnarray}
\begin{array}{r r r l@{\ }}
	& \hspace{.3in} \max_{\Lambda_*=\Lambda_*^T\in {\cal L}_{*}^0}
	\tr(\Lambda_*
	\begin{bmatrix}
		\bU & \bU^{\perp}
	\end{bmatrix}^T(X+W)\begin{bmatrix}
		\bU & \bU^{\perp}
	\end{bmatrix}) & \leq & \ell_1^*(X)\\
	\Longleftrightarrow \hspace{.0in} $ $ & \max_{\Lambda_{*,i}=\Lambda_{*,i}^T,\Lambda_{*,i}\Lambda_{*,i}^T\leq I,i\in\{1,2\}}
	\tr(\Lambda_{*,1}\bU^T X\bU +\Lambda_{*,2}(\bU^{\perp})^T W\bU^{\perp}) & \leq & \ell_1^*(X) \\
	\Longleftrightarrow \hspace{.0in} $ $ & \max_{\Lambda_{*,i}=\Lambda_{*,i}^T,\Lambda_{*,i}\Lambda_{*,i}^T\leq I,i\in\{1,2\}}
	\tr(\Lambda_{*,1}\Lambda_x +\Lambda_{*,2}(\bU^{\perp})^T W\bU^{\perp}) & \leq & \tr(\Lambda_x).
\end{array}\nonumber \\ \label{eq:genmcproof6}
\end{eqnarray}
Now, if at least one of the elements on the main diagonal of $\Lambda_{*,1}$, $\diag(\Lambda_{*,1})$, is smaller than 1, then the corresponding element on the diagonal of $\Lambda_x$ can be made arbitrarily large compared to the other elements of $\Lambda_x$ and one would have
\begin{eqnarray}
\begin{array}{r r r l@{\ }}
	& \max_{\Lambda_{*,i}=\Lambda_{*,i}^T,\Lambda_{*,i}\Lambda_{*,i}^T\leq I,i\in\{1,2\}}
	\tr(\Lambda_{*,1}\Lambda_x +\Lambda_{*,2}(\bU^{\perp})^T W\bU^{\perp}) & < & \tr(\Lambda_x) \\
	\Longleftrightarrow \hspace{.3in} $ $   & \max_{\Lambda_*=\Lambda_*^T\in {\cal L}_{*}^0}
	\tr(\Lambda_*
	\begin{bmatrix}
		\bU & \bU^{\perp}
	\end{bmatrix}^T(X+W)\begin{bmatrix}
		\bU & \bU^{\perp}
	\end{bmatrix}) & < & \ell_1^*(X)\\
	\Longleftrightarrow \hspace{.3in} $ $   & \max_{\Lambda_*=\Lambda_*^T\in {\cal L}_{*}}
	\tr(\Lambda_*
	\begin{bmatrix}
		\bU & \bU^{\perp}
	\end{bmatrix}^T(X+W)\begin{bmatrix}
		\bU & \bU^{\perp}
	\end{bmatrix}) & < & \ell_1^*(X),
\end{array}\nonumber \\ \label{eq:genmcproof7}
\end{eqnarray}
where the last equivalence holds since the difference of the terms on the left-hand side in the last two inequalities is bounded independently of $X$. Also, the last inequality in (\ref{eq:genmcproof7}) together with the first equality in (\ref{eq:genmcproof1a}) and the first inequality in (\ref{eq:genmcproof1}) produces (\ref{eq:genmcproof4}). Therefore the only scenario that is left as potentially not producing (\ref{eq:genmcproof4}) is when all the elements on the main diagonal are larger than or equal to 1. However, the two lemmas preceding the theorem show that in such a scenario ${\cal L}^0={\cal L}$ and one consequently has an equality instead of the inequality in (\ref{eq:genmcproof1a}) which then, together with
(\ref{eq:genmcproof1}), implies (\ref{eq:genmcproof4}). This completes the proof of the second (``the only if") part of the theorem and therefore of the entire theorem.
\end{proof}

The above derivations can be repeated for asymmetric matrices $X_{sol}$ as well. The final result differs in a cosmetic change. Namely, the right $\bar{U}$ should be replaced by a corresponding $\bar{V}$. Also, the condition in the theorem relates matrix $W$ to the null-space of matrix $A$ and as such belongs to the class of so-called vector matrix terminology (VMT) based conditions. In the MC and C-inf cases that are of our interest here, it is more convenient to deal with its a masking matrix terminology (MMT) analogue. Recalling on the proof of Theorem \ref{thm:thm1} and the role of matrix $W$ within that proof, one has that stating that $W$ belongs to the null-space of $A$ is basically equivalent to stating that $M\circ W=\0_{n\times n}$. In other words, one has the equivalence between the following two sets
\begin{equation}
(W\in \mR^{n\times n} | A\vecw(W)=\0_{m\times 1},W\neq \0_{n\times n})
\Longleftrightarrow (W\in \mR^{n\times n} | M\circ W=\0_{n\times n},W\neq \0_{n\times n}).
\label{eq:cinfanl1}
\end{equation}

Continuing further in the spirit of the RDT, the following corollary of the above theorem can be established as well.
\begin{corollary}(\textbf{\bl{$\ell_0^*-\ell_1^*$-equivalence condition}} -- \textbf{general}  $X$)
Assume the setup of Proposition \ref{prop:cauinf} with $X_{sol}$ being the unique solution of (\ref{eq:genmcl0posmmt}). Let the masking matrix $M\in\mR^{n\times n}$ have $m$ ones and $(n^2-m)$ zeros and let $A$ be generated via $M$, i.e. let $A$ be the matrix obtained after removing all the zero rows from $\diag^{-1}(\vecw(M))I_{n^2\times n^2}$. If and only if
\begin{equation}
\min_{W,W^TW=1,M\circ W=\0_{n\times n}}  \tr(\bU^TW\bV)+\ell_1^*((\bU^{\perp})^TW\bV^{\perp})\geq 0,
\label{eq:cinfcor1}
\end{equation}
then
\begin{equation}
\ell_0^*\Longleftrightarrow \ell_1^* \quad \mbox{and}\quad  \textbf{\emph{RMSE}}=\|\hat{X}-X_{sol}\|_F=0,\label{eq:cinfcor1a}
\end{equation}
and the solutions of (\ref{eq:genmcl0posmmt}) and (\ref{eq:genmcl1posmmt}) coincide.
\label{cor:cinfcor1}
\end{corollary}
\begin{proof}
Follows immediately as a combination of \ref{eq:genmcposthmcond1}, \ref{eq:genmcposthmcond2}, and (\ref{eq:cinfanl1}).
\end{proof}

\noindent \textbf{Remark:} Carefully comparing the conditions in (\ref{eq:genmcposthmcond1}) and (\ref{eq:cinfcor1}) one can observe that a strict inequality is loosened up a bit at the expense of the uniqueness assumption. With a little bit of extra effort one may avoid this. However, to make writings below substantially easier we will work with a non-strict inequality.

One also has the following alternative to the above corollary.

\begin{corollary}(\textbf{\bl{$\ell_0^*-\ell_1^*$-equivalence condition via masking matrix}} -- \textbf{general}  $X$)
Assume the setup of Proposition \ref{prop:cauinf} and Corollary \ref{cor:cinfcor1}. Let $I^{(l)}$ be as in (\ref{eq:cinfanl2a}). Then
\begin{center}
\tcbset{beamer,lower separated=false, fonttitle=\bfseries,
	coltext=black , interior style={top color=orange!10!yellow!30!white, bottom color=yellow!80!yellow!50!white}, title style={left color=orange!10!cyan!30!blue, right color=green!70!blue!20!black}}
\begin{tcolorbox}[beamer,title=\textbf{\yellow{$\ell_0^*\Longleftrightarrow \ell_1^* \quad \mbox{and}\quad  \textbf{\emph{RMSE}}=\|\hat{X}-X_{sol}\|_F=0$}},fonttitle=\bfseries,width=.98\linewidth]
	\begin{equation}
		\mbox{ If and only if} \quad \exists \Lambda| \Lambda^T\Lambda\leq I \quad \mbox{and} \quad (I^{(l)})^T \lp \bV\bU^T+\bV^{\perp}\Lambda(\bU^{\perp})^T\rp I^{(l)}=0.\label{eq:cinfcor2eq1}
	\end{equation}
\end{tcolorbox}
\end{center}
\label{cor:cinfcor2}
\end{corollary}
\begin{proof}
The ``if" part follows from Corollary \ref{cor:cinfcor1} and (\ref{eq:cinfanl5a}). The ``only if" part follows after noting that all the inequalities in (\ref{eq:cinfanl2})-(\ref{eq:cinfanl5a}) are written for generic instructional purposes. Due to the underlying convexity and the strong duality they all actually can be replaced with equalities as well.
\end{proof}

Proposition \ref{thm:genmcthmregposcond} and Corollaries \ref{cor:cinfcor1} and \ref{cor:cinfcor2} consider the so-called idealized low rank scenario. In other words, they consider scenarios where the underlying matrix $X_{sol}$ is of rank $k<n$. Corollary \ref{cor:cinfcor1} particulary states that in such a scenario one has that the nuclear norm optimization in (\ref{eq:genmcl1posmmt}) (when it comes to its equivalence with the corresponding $\ell_0^*$ from (\ref{eq:genmcl0posmmt})) effectively boils down to the residual optimization problem given in (\ref{eq:cinfcor1}). Also, in the idealized context, it is possible to find scenarios when the optimal $W$ in (\ref{eq:cinfcor1}) is zero. In the non-idealized ones, i.e. when $X_{sol}$ has full rank that is not possible and one has the residual error $W=\hat{X}-X_{sol}\neq 0$. The following theorem (basically Theorem \ref{thm:thm1}) is the key upgrade of the above corollaries from idealized low rank to approximately low rank scenario, that ultimately characterizes the residual error.

\begin{theorem}(Algebraic characterization of $W$)
Consider a $\bU\in\mR^{n\times k}$ such that $\bU^T\bU=I_{k\times k}$ and a $\bV\in\mR^{n\times k}$ such that $\bV^T\bV=I_{k\times k}$ and an approximately  rank $k$ matrix $X_{sol}=X\in\mR^{n\times n}$, such that $X_{sol}=U\bar{\sigma}V^T, \bar{\sigma}=\diag(\sigma,\epsilon_\sigma),\sigma\in\mR^k,\epsilon_\sigma\in\mR^{n-k},\sigma\gg \sigma_\epsilon,\epsilon_\sigma\leq \sigma_\epsilon$.
Also, let the orthogonal spans $\bU^{\perp}\in\mR^{n\times (n-k)}$ and $\bV^{\perp}\in\mR^{n\times (n-k)}$ be such that $U\triangleq \begin{bmatrix}
\bU & \bU^{\perp}
\end{bmatrix}$ and $V\triangleq \begin{bmatrix}
\bV & \bV^{\perp}
\end{bmatrix}$ and
\begin{equation}\label{eq:cinfthm0}
U^TU\triangleq \begin{bmatrix}
	\bU & \bU^{\perp}
\end{bmatrix}^T\begin{bmatrix}
	\bU & \bU^{\perp}
\end{bmatrix}=I_{n\times n} \quad \mbox{and} \quad
V^TV \triangleq\begin{bmatrix}
	\bV & \bV^{\perp}
\end{bmatrix}^T\begin{bmatrix}
	\bV & \bV^{\perp}
\end{bmatrix}=I_{n\times n}.
\end{equation}
With $M\in\mR^{n\times n}$ as in (\ref{eq:cinfanl2a}), assume that $Y= M\circ X_{sol}$ and let $\hat{X}$ be the solution of (\ref{eq:genmcl1posmmt}). Set
\begin{equation}
\hat{W} = \arg \min_{W,M\circ W=\0_{n\times n}}  \tr(\bU^TW\bV)+\ell_1^*(\diag(\epsilon_\sigma)+(\bU^{\perp})^TW\bV^{\perp}).
\label{eq:Acinfcor1}
\end{equation}
Then for any $X_{sol}$ that satisfies the above setup one has
\begin{equation}
\textbf{\emph{RMSE}}_{(ns)}=\|\hat{X}-X_{sol}\|_F\leq\|\hat{W}\|_F.\label{eq:Acinfcor1a}
\end{equation}
Moreover, there exists an $X_{sol}$ such that
\begin{equation}
\textbf{\emph{RMSE}}_{(ns)}=\|\hat{X}-X_{sol}\|_F=\|\hat{W}\|_F.\label{eq:Acinfcor1a1}
\end{equation}
\label{thm:thmA1}
\end{theorem}
\begin{proof}
Follows directly from Corollary \ref{cor:cinfcor1} by repeating all arguments and accounting for an additional term $\bar{U}^\perp\epsilon_\sigma( \bar{V}^\perp)^T$ in the representation of $X_{sol}$ (that term is propagated as an addition to $W$ and eventually appears only in the second term of (\ref{eq:Acinfcor1})). Effectively, (\ref{eq:Acinfcor1a}) states that the residual error will be no worse (larger in $\|\cdot\|_F$ norm) than the solution of the optimization in (\ref{eq:Acinfcor1}). Moreover, the tightness of the analyses leading to Proposition \ref{thm:genmcthmregposcond} and Corollary \ref{cor:cinfcor1} ensures that such an error is also achievable (as stated in (\ref{eq:Acinfcor1a1})).
\end{proof}

\section{Proof of Lemma \ref{lemma:lemma1}}

\subsection{Basics of free probability theory (FPT) -- random matrix variables}
\label{sec:fptprelmatrices}

We find it useful to first recall on some FPT basics. To start things off we consider two symmetric matrices, $A=A^T\in\mR^{n\times n}$ and $B=B^T\in\mR^{n\times n}$, as random variables. We also assume large $n$ regime and that the eigenspaces of these matrices are Haar distributed. Moreover, we will assume that their individual respective spectral laws are $f_A(\cdot)$ and $f_B(\cdot)$. Furthermore, we recall on the so-called Stieltjes transform (or as we will often call it G-transform) of a pdf $f(\cdot)$
\begin{eqnarray}
G(z) & \triangleq & \int_{I_f} \frac{f(x)}{z-x} dx, \quad z\in\mC\setminus I_f,  \label{eq:typwcanl6}
\end{eqnarray}
where $I_f$ is the domain of $f(\cdot)$. One then also has the inverse relation (somewhat analogous to the above relation between the inverse Fourier and the underlying pdf of the sum of scalar random variables)
\begin{eqnarray}
f(x) =  \lim_{\epsilon\rightarrow 0^+} \frac{G(x-i\epsilon)-G(x+i\epsilon)}{2i\pi}
\quad \mbox{or} \quad    f(x) =  -\lim_{\epsilon\rightarrow 0^+} \frac{\mbox{imag}(G(x+i\epsilon))}{\pi}.   \label{eq:typwcanl7}
\end{eqnarray}
For the above to hold it makes things easier to implicitly assume that $f(x)$ is continuous. We will, however, utilize it even in discrete (or semi-discrete) scenarios since the obvious asymptotic translation to continuity would make it fully rigorous. A bit later though, when we see some concrete examples where things of this nature may appear, we will say a few more words and explain more thoroughly what exactly can
be discrete and how one can deal with such a discreteness. In the meantime we proceed with general principles not necessarily worrying about all the underlying technicalities that may appear in scenarios deviating from the typically seen ones and potentially requiring additional separate addressing. To that end we continue by considering the $R(\cdot)$- and $S(\cdot)$-transforms that satisfy the following
\begin{eqnarray}
R(G(z))+\frac{1}{G(z)}=z,  \label{eq:typwcanl8}
\end{eqnarray}
and
\begin{eqnarray}
S(z)=\frac{1}{R(zS(z))} \quad \mbox{and}\quad R(z)=\frac{1}{S(zR(z))}.\label{eq:typwcanl9}
\end{eqnarray}
Let $f_A(\cdot)$ and $f_B(\cdot)$ be the spectral distributions of $A$ and $B$ and let $R_A(z)$/$S_A(z)$ and $R_B(z)$/$S_B(z)$ be their associated $R(\cdot)$-/$S(\cdot)$-transforms. One then has the following
\begin{center}
\tcbset{beamer,lower separated=false, fonttitle=\bfseries,
coltext=black , interior style={top color=orange!10!yellow!30!white, bottom color=yellow!80!yellow!50!white}, title style={left color=black, right color=red!70!orange!30!blue}}
\begin{tcolorbox}[beamer,title=\textbf{Key Voiculescu's FPT concepts \cite{Voic86,Voic87}:},fonttitle=\bfseries,width=5in]
\vspace{-.1in}\begin{eqnarray}
	\begin{array}{r c l l r c l}
		C & = & A+B & \quad \Longrightarrow \quad $ $ & R_C(z) & = & R_A(z)+R_B(z)\\
		C & = & AB  & \quad \Longrightarrow \quad $ $ & S_C(z) & = & S_A(z)S_B(z).
	\end{array}\label{eq:typwcanl10}
\end{eqnarray}
\end{tcolorbox}
\end{center}
\noindent Now it is relatively easy to see that (\ref{eq:typwcanl6})-(\ref{eq:typwcanl10}) are sufficient to determine the spectral distribution of the sum or the product of two independent matrices with given spectral densities and the Haar distributed bases of eigenspaces. The above is, of course, a generic principle. It can be applied pretty much always as long as one has access to the statistics of the underlying matrices $A$ and $B$. In the following section, we raise the bar a bit higher and show that, in the cases of our interest in this paper, one can use all of the above in such a manner that eventually all the relevant quantities are explicitly determined. Moreover, although the methodology may, on occasion, seem a bit involved, the final results will turn out to be presentable in fairly neat and elegant closed forms.

\subsection{Spectral considerations}
\label{sec:fptprelmatrices1}

We start with a trivial observation. Let $f_{\calV}(\cdot)$ be the spectral distribution of $\calV$. Then
\begin{eqnarray}
f_\calV(x)=(1-\beta)\delta(1-x)+\beta\delta(x), \label{eq:typwcanl22}
\end{eqnarray}
where $\delta(\cdot)$ stands for the standard delta function with nonzero value only when its argument takes value zero. Using the definition of the $G$-transform from (\ref{eq:typwcanl6}) we can find
\begin{eqnarray}
G_\calV(z) = \int \frac{f_\calV(x)}{z-x} dx = \int \frac{(1-\beta)\delta(1-x)+\beta\delta(x)}{z-x} dx=
\frac{1-\beta}{z-1}+\frac{\beta}{z}=\frac{z-\beta}{z^2-z}.  \label{eq:typwcanl23}
\end{eqnarray}
Also, we have
\begin{eqnarray}
R_\calV(y)=z-\frac{1}{y} \quad \mbox{with} \quad y=G_\calV(z) \quad \mbox{and} \quad z=G_\calV^{-1}(y).  \label{eq:typwcanl24}
\end{eqnarray}
From (\ref{eq:typwcanl23}) and (\ref{eq:typwcanl24}) we further find
\begin{eqnarray}
G_\calV(z)=y \quad  \Longleftrightarrow  \quad \frac{z-\beta}{z^2-z}=y \quad \Longleftrightarrow \quad  z^2y-z(y+1)+\beta=0.  \label{eq:typwcanl25}
\end{eqnarray}
Solving for $z$ gives
\begin{eqnarray}
z=\frac{y+1\pm\sqrt{(y+1)^2-4\beta y}}{2y}.  \label{eq:typwcanl26}
\end{eqnarray}
Combining (\ref{eq:typwcanl24}) and (\ref{eq:typwcanl26}) we obtain for the $R$-transform
\begin{eqnarray}
R_\calV(y)=z-\frac{1}{y}= \frac{y-1\pm\sqrt{(y+1)^2-4\beta y}}{2y},  \label{eq:typwcanl27}
\end{eqnarray}
where we for the completeness adopt the strategy to keep both $\pm$ signs. To determine the $S$-transform we start by combining (\ref{eq:typwcanl9}) and (\ref{eq:typwcanl27})
\begin{eqnarray}
S_\calV(z)=\frac{1}{R_\calV(zS_\calV(z))}= \frac{1}{\frac{zS_\calV(z)-1\pm\sqrt{(zS_\calV(z)+1)^2-4\beta zS_\calV(z)}}{2zS_\calV(z)}}.  \label{eq:typwcanl28}
\end{eqnarray}
After a bit of algebraic transformations we have
{\small \begin{eqnarray}
\begin{array}{c r c l}
	& zS_\calV(z)-1-2z &  =  & \mp\sqrt{(zS_\calV(z)+1)^2-4\beta zS_\calV(z)} \\
	\Longleftrightarrow & (zS_\calV(z)-1-2z)^2 &  =  & (zS_\calV(z)+1)^2-4\beta zS_\calV(z) \\
	\Longleftrightarrow & (zS_\calV(z))^2-2(2z+1)zS_\calV(z)+(2z+1)^2 &  =  & (zS_\calV(z))^2+2zS_\calV+1-4\beta zS_\calV(z) \\
	\Longleftrightarrow & -2(2z+1)zS_\calV(z)+4z^2+4z &  =  &  2zS_\calV(z)-4\beta zS_\calV(z) \\
	\Longleftrightarrow & 4z^2+4z &  =  &  (4z^2+4z)S_\calV(z)-4\beta zS_\calV(z) \\
	\Longleftrightarrow & z+1 &  =  &  S_\calV(z)(z+1-\beta).
\end{array}\nonumber \\\label{eq:typwcanl29}
\end{eqnarray}}From (\ref{eq:typwcanl29}) we finally have
\begin{eqnarray}
S_\calV(z)=\frac{z+1}{z+1-\beta}.\label{eq:typwcanl30}
\end{eqnarray}
As this is a very generic result it is useful to have it formalized in the following lemma.
\begin{lemma}
Let $\bV^{\perp}\in\mR^{n\times (n-k)}$ be Haar distributed unitary basis of an $(n-k)$-dimensional subspace of $\mR^n$. Let $\calV$ be defined as
\begin{eqnarray}
{\cal V}  \triangleq  \bV^{\perp}(\bV^{\perp})^T.\label{eq:typwclemma1eq1}
\end{eqnarray}
In the large $n$ linear regime, with $\beta\triangleq\lim_{n\rightarrow\infty}\frac{k}{n}$, the $S$-transform of the spectral density of $\calV$, $f_\calV(\cdot)$, is
\begin{eqnarray}
S_\calV(z)=\frac{z+1}{z+1-\beta}.\label{eq:typwclemma1eq2}
\end{eqnarray}\label{lemma:typwclemma1}
\end{lemma}
\begin{proof}
Follows from the above discussion.
\end{proof}

Since $\calV$ and $\calU$  are  structurally identical (with the only difference being one of their dimensions) we easily have
\begin{eqnarray}
S_\calU(z)=\frac{z+1}{z+1-\eta}.\label{eq:typwcanl31}
\end{eqnarray}
A combination of (\ref{eq:typwcanl10}), (\ref{eq:typwcanl30}), and (\ref{eq:typwcanl31}) gives
\begin{eqnarray}
S_{\tilde{D}}(z)=\frac{(z+1)^2}{(z+1-\beta)(z+1-\eta)}.\label{eq:typwcanl32}
\end{eqnarray}
From (\ref{eq:typwcanl9}) we also have
\begin{equation}
R_{\tilde{D}}(z)=\frac{1}{S_{\tilde{D}}(zR_{\tilde{D}}(z))}=
\frac{1}{\frac{(zR_{\tilde{D}}(z)+1)^2}{(zR_{\tilde{D}}(z)+1-\beta)(zR_{\tilde{D}}(z)+1-\eta)}}=
\frac{(zR_{\tilde{D}}(z)+1-\beta)(zR_{\tilde{D}}(z)+1-\eta)}{(zR_{\tilde{D}}(z)+1)^2}.\label{eq:typwcanl33}
\end{equation}
Moreover, (\ref{eq:typwcanl7}) gives
\begin{eqnarray}
R_{\tilde{D}}(G_{\tilde{D}}(z))+\frac{1}{G_{\tilde{D}}(z)}=z,\label{eq:typwcanl34}
\end{eqnarray}
and
\begin{eqnarray}
G_{\tilde{D}}(z)R_{\tilde{D}}(G_{\tilde{D}}(z))=zG_{\tilde{D}}(z)-1,\label{eq:typwcanl35}
\end{eqnarray}
From (\ref{eq:typwcanl33}) one further finds
\begin{eqnarray}
R_{\tilde{D}}(G_{\tilde{D}}(z))=
\frac{(G_{\tilde{D}}(z)R_{\tilde{D}}(G_{\tilde{D}}(z))+1-\beta)(G_{\tilde{D}}(z)R_{\tilde{D}}(G_{\tilde{D}}(z))+1-\eta)}
{(G_{\tilde{D}}(z)R_{\tilde{D}}(G_{\tilde{D}}(z))+1)^2}.\label{eq:typwcanl36}
\end{eqnarray}
After plugging (\ref{eq:typwcanl35}) in (\ref{eq:typwcanl36}), we have
\begin{eqnarray}
R_{\tilde{D}}(G_{\tilde{D}}(z))&=&
\frac{(zG_{\tilde{D}}(z)-1+1-\beta)(zG_{\tilde{D}}(z)-1+1-\eta)}
{(zG_{\tilde{D}}(z)-1+1)^2}\nonumber \\
&=&\frac{(zG_{\tilde{D}}(z)-\beta)(zG_{\tilde{D}}(z)-\eta)}
{(zG_{\tilde{D}}(z))^2}.
\label{eq:typwcanl37}
\end{eqnarray}
A combination of (\ref{eq:typwcanl34}) and (\ref{eq:typwcanl37}) further gives
\begin{eqnarray}
z-\frac{1}{G_{\tilde{D}}(z)}=
\frac{(zG_{\tilde{D}}(z)-\beta)(zG_{\tilde{D}}(z)-\eta)}
{(zG_{\tilde{D}}(z))^2}.\label{eq:typwcanl38}
\end{eqnarray}
From (\ref{eq:typwcanl38})  we quickly find
\begin{eqnarray}
z^3(G_{\tilde{D}}(z))^2- z^2G_{\tilde{D}}(z)=
z^2(G_{\tilde{D}}(z))^2-(\beta+\eta)zG_{\tilde{D}}(z)+\beta\eta,\label{eq:typwcanl39}
\end{eqnarray}
and
\begin{eqnarray}
(G_{\tilde{D}}(z))^2(z^3-z^2)- G_{\tilde{D}}(z)(z^2-z(\beta+\eta))-\beta\eta=0.\label{eq:typwcanl40}
\end{eqnarray}
Solving for $G_{\tilde{D}}(z)$ finally gives
\begin{eqnarray}
G_{\tilde{D}}^{\pm}(z)=\frac{z^2-z(\beta+\eta)\pm\sqrt{(z^2-z(\beta+\eta))^2+4\beta\eta(z^3-z^2)}}{2(z^3-z^2)},\label{eq:typwcanl41}
\end{eqnarray}
or
\begin{eqnarray}
G_{\tilde{D}}^{\pm}(z)=\frac{z-(\beta+\eta)\pm\sqrt{(z-(\beta+\eta))^2+4\beta\eta(z-1)}}{2(z^2-z)}.\label{eq:typwcanl42}
\end{eqnarray}

The above is sufficient to establish the following lemma.
\begin{lemma}
Let $\bV^{\perp}\in\mR^{n\times (n-k)}$ and $\bU_D^{\perp}\in\mR^{n\times (n-k)}$ be Haar distributed unitary bases of $(n-k)$-dimensional subspaces of $\mR^n$. Let $\calV$, $\calU$, and $\tilde{D}$ be defined as
\begin{eqnarray}
{\cal V} & \triangleq & \bV^{\perp}(\bV^{\perp})^T \nonumber \\
{\cal U} & \triangleq & \bU_D^{\perp}(\bU_D^{\perp})^T\nonumber \\
\tilde{D} & \triangleq & {\cal V}{\cal U}.\label{eq:typwclemma1aeq1}
\end{eqnarray}
In the large $n$ linear regime, with $\beta\triangleq\lim_{n\rightarrow\infty}\frac{k}{n}$, the $G$-transform of the spectral density of $\tilde{D}$, $f_{\tilde{D}}(\cdot)$, is
\begin{eqnarray}
G_{\tilde{D}}^{\pm}(z)=\frac{z-(\beta+\eta)\pm\sqrt{(z-(\beta+\eta))^2+4\beta\eta(z-1)}}{2(z^2-z)}.\label{eq:typwclemma1aeq2}
\end{eqnarray}\label{lemma:typwclemma1a}
\end{lemma}
\begin{proof}
Follows from the above discussion. The ``$+$/$-$" signs are taken for negative/positive imaginary part under the root.
\end{proof}

One then relies on (\ref{eq:typwcanl7}) to determine $f_{\tilde{D}}(x)$ as
\begin{eqnarray}
f_{\tilde{D}}(x) =  -\lim_{\epsilon\rightarrow 0^+} \frac{\mbox{imag}(G_{\tilde{D}}(x+i\epsilon))}{\pi}.   \label{eq:typwcanl43}
\end{eqnarray}
The above is a generic procedure and we in Figure \ref{fig:cinfspecGplusG} show the results that one can get for two concrete values $\beta=0.2$ and $\eta=0.6$. One should note that it is not clear \emph{a priori} which of the two $\pm$ signs should be used. As Figure \ref{fig:cinfspecGplusG} indicates one most definitely has to be fairly careful and account for both signs. From Figure \ref{fig:cinfspecGplusG} one further observes that there are four critical points in the spectrum itself: the locations of the two delta functions, zero and one, and two edges of the spectrum's bulk, $x_l$ and $x_u$. The values of these points are shown in the plots on the right hand side. In general one can actually determine their closed forms as well. Moreover, it turns out that one can determine the closed form of the entire spectral function. The section that follows analyzes the spectrum of $\tilde{D}$ in more details and eventually provides the closed form expressions for all the relevant spectral features.

\begin{figure}[htb]
\centerline{\includegraphics[width=.8\columnwidth]{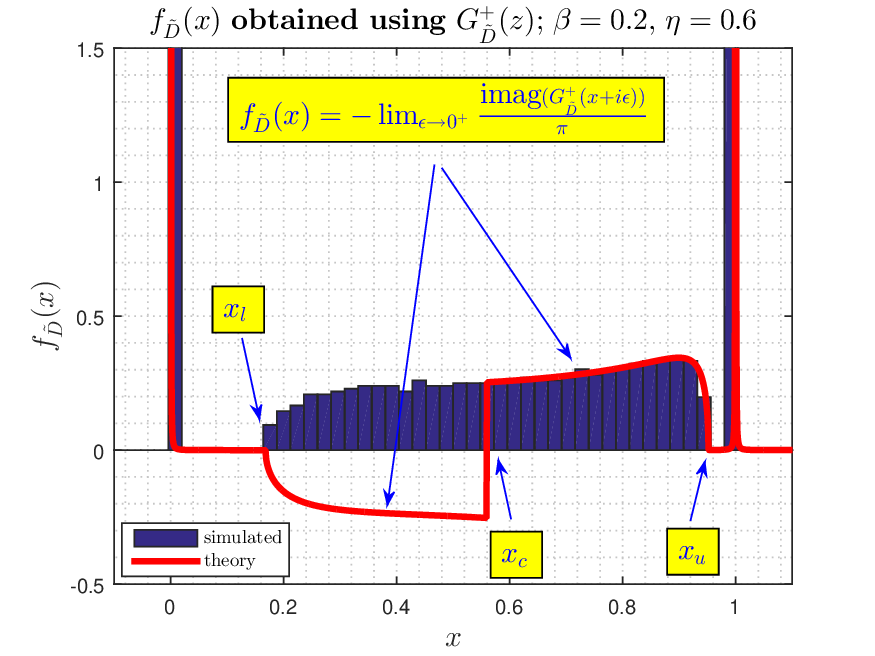}}
\caption{Both $G_{\tilde{D}}^+(z)$ and $G_{\tilde{D}}^-(z)$ need to be taken into account}
\label{fig:cinfspecGplusG}
\end{figure}
\begin{figure}[htb]
\centerline{\includegraphics[width=.8\columnwidth]{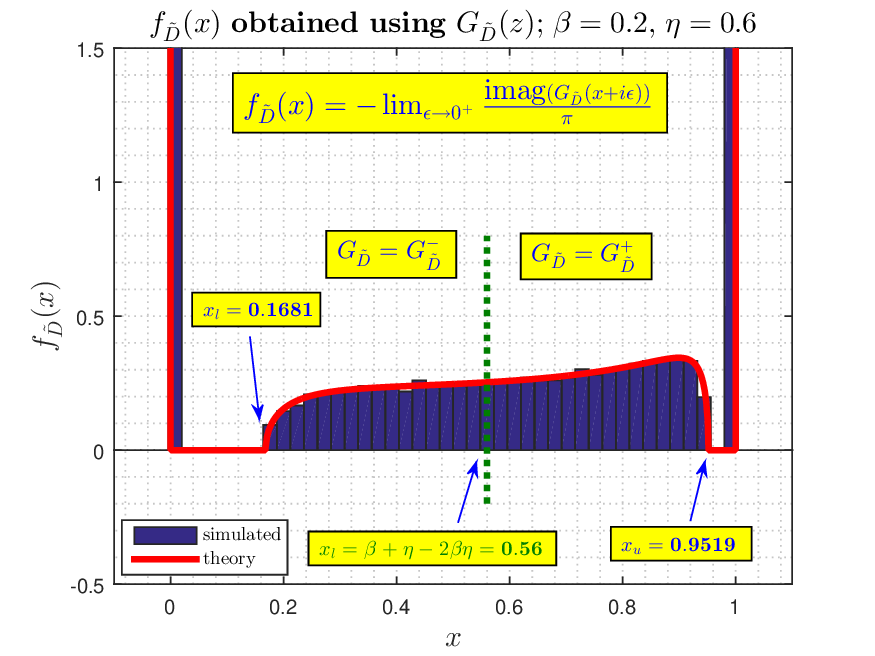}}
\caption{Both $G_{\tilde{D}}^+(z)$ and $G_{\tilde{D}}^-(z)$ need to be taken into account}
\label{fig:cinfspecGplusG}
\end{figure}

\subsubsection{The spectrum of $\tilde{D}$ -- closed form expressions}
\label{sec:fpteqvspecwctildeD}

As one of our main concerns in this paper is the utilization of the final results that we will get in this section and not necessarily the presentation of the tiny details needed to get them, we will sketch all the key arguments and leave out all the unnecessary minute details. However, we do emphasize that the sketch will contain all the key pointers so that with a little bit of effort one, if in a need, can fill in all the missing pieces of the overall mosaic.

As mentioned above, looking at the denominator of (\ref{eq:typwcanl42}) and keeping in mind the $f\longleftrightarrow G$ connection from (\ref{eq:typwcanl43}) one observes that the pdf of interest, $f_{\tilde{D}}(x)$, potentially has two delta functions, one at zero and the other one at one. Moreover, the bulk of the spectrum will be in the range where the real part under the root is negative. It also goes almost without saying that the entire spectrum will be located between zero and one. Finally, the breaking point, $x_c$, where one needs to switch from $G_{\tilde{D}}^+(z)$ to $G_{\tilde{D}}^-(z)$ in (\ref{eq:typwcanl43}) is determined as the value where the imaginary part under the root changes its sign. Equipped with these observations one can then proceed to actually concretely determine some of the relevant quantities.

Based on what we have just observed above, we first express $f_{\tilde{D}}(x)$ as the sum of its three key constitutive parts (two delta functions and the bulk)
\begin{eqnarray}
f_{\tilde{D}}(x) = f_0\delta (x-0) +f_{\tilde{D}}^{(b)}(x)+f_1\delta(x-1) .   \label{eq:typwcanl44}
\end{eqnarray}
From (\ref{eq:typwcanl44}) one has that $f_{\tilde{D}}(x)$ will be fully specified if one can determine the delta multipliers $f_0$ and $f_1$, and the bulk pdf $f_{\tilde{D}}^{(b)}(\cdot)$.

\underline{\bl{\textbf{ 1) Finding $f_0$:}}} To determine $f_0$ we start by observing from (\ref{eq:typwcanl43}) for $x=0$
\begin{eqnarray}
f_{\tilde{D}}(0) =  -\lim_{\epsilon\rightarrow 0^+} \frac{\mbox{imag}(G_{\tilde{D}}(i\epsilon))}{\pi}.   \label{eq:typwcanl45}
\end{eqnarray}
Utilizing (\ref{eq:typwcanl42}) we further have
\begin{eqnarray}
f_{\tilde{D}}(0) &  =  & -\frac{1}{\pi}\lim_{\epsilon\rightarrow 0^+} \mbox{imag}\lp \frac{i\epsilon-(\beta+\eta)\pm\sqrt{(i\epsilon-(\beta+\eta))^2+4\beta\eta(i\epsilon-1)}}{2((i\epsilon)^2-i\epsilon)}\rp \nonumber \\
&  =  & -\frac{1}{\pi}\lim_{\epsilon\rightarrow 0^+} \mbox{imag}\lp \frac{-(\beta+\eta)\pm\sqrt{-\epsilon^2+(\beta+\eta)^2-4\beta\eta -2i\epsilon (\beta+\eta-2\beta\eta)}}{2(-\epsilon^2-i\epsilon)}\rp \nonumber \\
&  =  & -\frac{1}{\pi}\lim_{\epsilon\rightarrow 0^+} \mbox{imag}\lp \frac{-(\beta+\eta)\pm\sqrt{(\beta-\eta)^2  }}{2(-i\epsilon)}\rp \nonumber \\
&  =  & -\frac{1}{\pi}\lim_{\epsilon\rightarrow 0^+} \mbox{imag}\lp \frac{-(\beta+\eta)-|\beta-\eta| }{-2i\epsilon}\rp \nonumber \\
&  =  & \lp \beta+\eta+|\beta-\eta|  \rp \lp -\frac{1}{\pi}\lim_{\epsilon\rightarrow 0^+} \mbox{imag}\lp \frac{1}{i\epsilon}\rp \rp\nonumber \\
&  =  & \max(\beta,\eta)\delta(0),   \label{eq:typwcanl46}
\end{eqnarray}
where the fourth equality (the choice of the ``$-$" sign in $\pm$) follows since $0\leq x_c$ (the spectrum belongs to the interval $[0,1]$ and $x_c$ must be in the spectrum) and the last equality follows since by convention
\begin{eqnarray}
\delta(0)=\lp -\frac{1}{\pi}\lim_{\epsilon\rightarrow 0^+} \mbox{imag}\lp \frac{1}{i\epsilon}\rp \rp.   \label{eq:typwcanl47}
\end{eqnarray}
To see the rationale behind (\ref{eq:typwcanl47}) we briefly digress and start with
\begin{equation}
g(x)=\delta(x). \label{eq:typwcanl48}
\end{equation}
Then from (\ref{eq:typwcanl6})
\begin{equation}
G(z)=\int_x \frac{\delta(x)dx}{z-x}=\frac{1}{z}, \label{eq:typwcanl49}
\end{equation}
and from (\ref{eq:typwcanl7})
\begin{equation}
\delta(x)=-\lim_{\epsilon\rightarrow 0^+}\frac{\mbox{imag}(G(x+i\epsilon))}{\pi}. \label{eq:typwcanl50}
\end{equation}
For $x=0$ then
\begin{equation}
\delta(0)=-\lim_{\epsilon\rightarrow 0^+}\frac{\mbox{imag}(G(i\epsilon))}{\pi}
=-\lim_{\epsilon\rightarrow 0^+}\mbox{imag}\lp\frac{1}{\pi i\epsilon}\rp
=-\frac{1}{\pi}\lim_{\epsilon\rightarrow 0^+}\mbox{imag}\lp\frac{1}{i\epsilon}\rp, \label{eq:typwcanl51}
\end{equation}
which is identical to (\ref{eq:typwcanl47}). The above description of the delta function may not necessarily be the most adequate one. However, for what we need here it is conceptually sufficient. Namely, we are here interested in determining the proportionality constants that multiply the delta functions rather than the functions' expressions themselves. One way to make everything more adequate would be to translate everything into the continuous domain by choosing a continuous function as an asymptotic replacement for $\delta(x)$. For example, one can use the Gaussian continual approximation
\begin{equation}
\delta(x)=\lim_{\sigma\rightarrow 0^+}\frac{e^{-\frac{x^2}{2\sigma^2}}}{\sqrt{2\pi\sigma^2}}. \label{eq:typwcanl52}
\end{equation}
Then all the above holds for small $\sigma=\epsilon\sqrt{\pi}/2$ and
\begin{equation}
\delta(x)\rightarrow \lim_{\sigma\rightarrow 0^+}\frac{e^{-\frac{x^2}{2\sigma^2}}}{\sqrt{2\pi\sigma^2}}
\rightarrow \lim_{\sigma\rightarrow 0^+}\frac{e^{-\frac{x^2}{\pi\epsilon^2}}}{\pi\epsilon}\quad \mbox{and} \quad
\delta(0)\rightarrow \lim_{\sigma\rightarrow 0^+}\frac{1}{\pi\epsilon}. \label{eq:typwcanl53}
\end{equation}
While computing and manipulating the transforms mentioned, it is necessary to consider the minor $\epsilon$-differences they induce. These differences are conceptually insignificant, and our results remain valid for small values of $\sigma$ or $\epsilon$. Including these $\epsilon$-modifications in our writing would lead to a more cumbersome presentation, demanding the addition of many minor details to demonstrate their marginal impact. Since, on the other hand, they contribute exactly nothing to the essence of the arguments and final results we chose to operate in a semi-discrete domain with the delta functions. As a consequence one has the expressions given in (\ref{eq:typwcanl47}) and (\ref{eq:typwcanl51}). We believe that a little bit of conventional inadequacy is better than to overwhelm the presentation with details which would make the overall content less accessible. 

\underline{\bl{\textbf{ 2) Finding $f_1$:}}} To determine $f_1$ we follow the above methodology and start by observing from (\ref{eq:typwcanl43}) for $x=1$
\begin{eqnarray}
f_{\tilde{D}}(1) =  -\lim_{\epsilon\rightarrow 0^+} \frac{\mbox{imag}(G_{\tilde{D}}(1+i\epsilon))}{\pi}.   \label{eq:typwcanl54}
\end{eqnarray}
Further utilization of (\ref{eq:typwcanl42}) gives
{\small \begin{eqnarray}
f_{\tilde{D}}(1) &  =  & \-\frac{1}{\pi}\lim_{\epsilon\rightarrow 0^+} \mbox{imag}\lp \frac{1+i\epsilon-(\beta+\eta)\pm\sqrt{(1+i\epsilon-(\beta+\eta))^2+4\beta\eta i\epsilon}}{2((1+i\epsilon)^2-1-i\epsilon)}\rp \nonumber \\
&  =  & -\frac{1}{\pi}\lim_{\epsilon\rightarrow 0^+} \mbox{imag}\lp \frac{1-(\beta+\eta)\pm\sqrt{-\epsilon^2+(1-(\beta+\eta))^2 -2i\epsilon (-1+\beta+\eta-2\beta\eta)}}{2(-\epsilon^2+i\epsilon)}\rp \nonumber \\
&  =  & -\frac{1}{\pi}\lim_{\epsilon\rightarrow 0^+} \mbox{imag}\lp \frac{1-(\beta+\eta)\pm\sqrt{(1-(\beta-\eta))^2  }}{2(i\epsilon)}\rp \nonumber \\
&  =  & -\frac{1}{\pi}\lim_{\epsilon\rightarrow 0^+} \mbox{imag}\lp \frac{1-(\beta+\eta)+|1-(\beta+\eta)| }{2i\epsilon}\rp \nonumber \\
&  =  & \lp 1-(\beta+\eta)+|1-(\beta+\eta)|  \rp \lp -\frac{1}{\pi}\lim_{\epsilon\rightarrow 0^+} \mbox{imag}\lp \frac{1}{i\epsilon}\rp \rp\nonumber \\
&  =  & \max(1-(\beta+\eta),0)\delta(0),   \label{eq:typwcanl55}
\end{eqnarray}}where the fourth equality (the choice of the ``$+$" sign in $\pm$) follows since now $x_c\leq 1$ and the last equality follows by the above discussed $\delta(0)$ convention.

\underline{\bl{\textbf{ 3) Finding $f_{\tilde{D}}^{(b)}(x)$:}}} To determine $f_{\tilde{D}}^{(b)}(x)$ for $x\notin \{0,1\}$ we again start with (\ref{eq:typwcanl43}) and, for a general $x$, we write the following 
\begin{eqnarray}
f_{\tilde{D}}^{(b)}(x) =  -\lim_{\epsilon\rightarrow 0^+} \frac{\mbox{imag}(G_{\tilde{D}}(x+i\epsilon))}{\pi}.   \label{eq:typwcanl56}
\end{eqnarray}
Relying once again on (\ref{eq:typwcanl42}) we, for $x\notin \{0,1\}$, have
{\tiny \begin{eqnarray}
f_{\tilde{D}}^{(b)}(x) &  =  & -\frac{1}{\pi}\lim_{\epsilon\rightarrow 0^+} \mbox{imag}\lp \frac{x+i\epsilon-(\beta+\eta)\pm\sqrt{(x+i\epsilon-(\beta+\eta))^2+4\beta\eta (x+i\epsilon-1}}{2((x+i\epsilon)^2-x-i\epsilon)}\rp \nonumber \\
&  =  & -\frac{1}{\pi}\lim_{\epsilon\rightarrow 0^+} \mbox{imag}\lp \frac{x-(\beta+\eta)\pm\sqrt{-\epsilon^2+(x-(\beta+\eta))^2+4\beta\eta(x-1)-2i\epsilon (-x+\beta+\eta-2\beta\eta)}}{2(x^2-x-\epsilon^2+i\epsilon(2x-1))}\rp \nonumber \\
&  =  & -\frac{1}{\pi}\lim_{\epsilon\rightarrow 0^+} \mbox{imag}\lp \frac{x-(\beta+\eta)\pm\sqrt{(x-(\beta+\eta))^2+4\beta\eta(x-1)-2i\epsilon (-x+\beta+\eta-2\beta\eta)}}{2(x^2-x)}\rp \nonumber \\
&  =  & -\frac{1}{\pi}\lim_{\epsilon\rightarrow 0^+} \mbox{imag}\lp \frac{\pm\sqrt{(x-(\beta+\eta))^2+4\beta\eta(x-1)-2i\epsilon (-x+\beta+\eta-2\beta\eta)}}{2(x^2-x)}\rp.   \label{eq:typwcanl56}
\end{eqnarray}}
Now, since one is interested in the imaginary part of interest is the region of $x$ where the real part under the root is negative (outside that region, i.e in the region of $x$ where the real part under the root is nonnegative $f_{\tilde{D}}^{(b)}(x)$ is zero). To determine the region of interest we start by setting
\begin{eqnarray}
T_{\tilde{D}}  \triangleq   \{x\in\mR| (x-(\beta+\eta))^2+4\beta\eta(x-1) \leq 0\} \quad \mbox{and}\quad
x_c\triangleq  \beta+\eta-2\beta\eta.
\label{eq:typwcanl57}
\end{eqnarray}
To explicitly characterize $T_{\tilde{D}}$ we look at the following
\begin{eqnarray}
\begin{array}{r r c l}
$ $  & (x-(\beta+\eta))^2+4\beta\eta(x-1) & = &  0  \\
\Longleftrightarrow  &  x^2-2x(\beta+\eta-2\beta\eta)+(\beta+\eta)^2-4\beta\eta & = & 0 \\
\Longleftrightarrow  &  x^2-2x(\beta+\eta-2\beta\eta)+(\beta-\eta)^2 & = & 0.
\end{array}\label{eq:typwcanl58}
\end{eqnarray}
Solving for $x$ one finds
\begin{eqnarray}
x &  = &   \frac{2(\beta+\eta-2\beta\eta)\pm\sqrt{(2(\beta+\eta-2\beta\eta))^2-4(\beta-\eta)^2}}{2} \nonumber \\
& = & \beta+\eta-2\beta\eta \pm \sqrt{(\beta+\eta-2\beta\eta)^2-(\beta-\eta)^2}.
\label{eq:typwcanl59}
\end{eqnarray}
Setting
\begin{eqnarray}
x_l & \triangleq & \beta+\eta-2\beta\eta  -  \sqrt{(\beta+\eta-2\beta\eta)^2-(\beta-\eta)^2} \nonumber \\
x_u & \triangleq & \beta+\eta-2\beta\eta  +  \sqrt{(\beta+\eta-2\beta\eta)^2-(\beta-\eta)^2},
\label{eq:typwcanl60}
\end{eqnarray}
one has
\begin{eqnarray}
T_{\tilde{D}}=\{x\in\mR| x\in [x_l,x_u]\}.
\label{eq:typwcanl61}
\end{eqnarray}
Moreover, from (\ref{eq:typwcanl60}), one also has
\begin{eqnarray}
0\leq x_l \leq x_u\leq 1,
\label{eq:typwcanl62}
\end{eqnarray}
with
\begin{eqnarray}
x_l=0 \quad \mbox{if} \quad \beta=\eta \quad \mbox{and} \quad x_u=1  \quad \mbox{if} \quad \beta=\eta=0.5.
\label{eq:typwcanl63}
\end{eqnarray}
The first two inequalities in (\ref{eq:typwcanl60}) are trivial, whereas the third one follows after noting
\begin{eqnarray}
\beta+\eta-2\beta\eta\leq \max(\beta,1-\beta)\leq 1,
\label{eq:typwcanl64}
\end{eqnarray}
and observing the following sequence
\begin{eqnarray}
\begin{array}{r r c l}
&  \beta+\eta-2\beta\eta  +  \sqrt{(\beta+\eta-2\beta\eta)^2-(\beta-\eta)^2} & \leq & 1 \\
\Longleftrightarrow   &    (\beta+\eta-2\beta\eta)^2-(\beta-\eta)^2 & \leq & (1-(\beta+\eta-2\beta\eta))^2 \\
\Longleftrightarrow   &    -(\beta-\eta)^2 & \leq & 1-2(\beta+\eta-2\beta\eta) \\
\Longleftrightarrow   &    -(\beta-\eta)^2+2(\beta+\eta-2\beta\eta)-1 & \leq & 0 \\
\Longleftrightarrow   &    -(\beta+\eta)^2+2(\beta+\eta)-1 & \leq & 0 \\
\Longleftrightarrow   &    -(1-(\beta+\eta))^2 & \leq & 0.
\end{array}
\label{eq:typwcanl65}
\end{eqnarray}
Returning to (\ref{eq:typwcanl56}) we further have for $x\in T_{\tilde{D}}=[x_l,x_u]$
\begin{eqnarray}
f_{\tilde{D}}^{(b)}(x)
&  =  & -\frac{1}{\pi}\lim_{\epsilon\rightarrow 0^+} \mbox{imag}\lp \frac{\pm\sqrt{(x-(\beta+\eta))^2+4\beta\eta(x-1)-2i\epsilon (-x+\beta+\eta-2\beta\eta)}}{2(x^2-x)}\rp \nonumber \\
&  =  & -\frac{1}{\pi}\lim_{\epsilon\rightarrow 0^+} \mbox{imag}\lp \frac{\pm\sqrt{(x-(\beta+\eta))^2+4\beta\eta(x-1)-2i\epsilon (x_c-x)}}{2(x^2-x)}\rp \nonumber \\
&  =  & -\frac{1}{\pi}\lim_{\epsilon\rightarrow 0^+} \mbox{imag}\lp \frac{i\sqrt{-(x-(\beta+\eta))^2-4\beta\eta(x-1)}}{2(x^2-x)}\rp,   \label{eq:typwcanl66}
\end{eqnarray}
where the ``$+$" plus sign is chosen if $x_c\leq x\leq x_u$ and the ``$-$" sign is chosen if $x_l\leq x\leq x_c$. Finally, from (\ref{eq:typwcanl66}) one easily finds
\begin{eqnarray}
f_{\tilde{D}}^{(b)}(x)   &  =  &
\frac{\sqrt{-(x-(\beta+\eta))^2-4\beta\eta(x-1)}}{2\pi(x-x^2)} \quad \mbox{ if } \quad x_l\leq x\leq x_u.
\label{eq:typwcanl68}
\end{eqnarray}

The above is then sufficient to completely characterize the spectral distribution $f_{\tilde{D}}(x)$. We summarize the results in the following lemma (basically a mirrored version of Lemma \ref{lemma:lemma1}).

\begin{lemma}\label{lemma:typwclemma2}
Assume large $n$ linear regime with
\begin{eqnarray}
\beta\triangleq \lim_{n\rightarrow \infty}\frac{k}{n}\quad \mbox{and} \quad \eta\triangleq \lim_{n\rightarrow \infty}\frac{l}{n}. \label{eq:typwclemma2eq1}
\end{eqnarray}
Let $\bV^{\perp}\in\mR^{n\times (n-k)}$ be a Haar distributed basis of an $n-k$-dimensional subspace of $R^n$. Analogously, let $\bU_D^{\perp}\in\mR^{n\times (n-k)}$ be a Haar distributed basis of an $n-l$-dimensional subspace of $R^n$. Moreover, let $\bV^{\perp}\in\mR^{n\times (n-k)}$ and $\bU_D^{\perp}\in\mR^{n\times (n-k)}$ be independent of each other. Also, let ${\cal V}$, ${\cal U}$, and $\tilde{D}$ be defined as
\begin{eqnarray}
{\cal V} & \triangleq & \bV^{\perp}(\bV^{\perp})^T \nonumber \\
{\cal U} & \triangleq & \bU_D^{\perp}(\bU_D^{\perp})^T\nonumber \\
\tilde{D} & \triangleq & {\cal V}{\cal U}. \label{eq:typwclemma2eq2}
\end{eqnarray}
Set $x_l$ and $x_u$ as in (\ref{eq:typwcanl60}), i.e.
\begin{eqnarray}
x_l & \triangleq & \beta+\eta-2\beta\eta  -  \sqrt{(\beta+\eta-2\beta\eta)^2-(\beta-\eta)^2} \nonumber \\
x_u & \triangleq & \beta+\eta-2\beta\eta  +  \sqrt{(\beta+\eta-2\beta\eta)^2-(\beta-\eta)^2},
\label{eq:typwclemma2eq3}
\end{eqnarray}
Then the limiting spectral distribution of $\tilde{D}$, $f_{\tilde{D}}(x)$, is
\begin{eqnarray}
f_{\tilde{D}}(x)  &   = &  f_0\delta(x) +  f_{\tilde{D}}^{(b)}(x) + f_1f_0\delta(x-1) \nonumber \\
&  =  &
\max(\beta,\eta)\delta(x) + f_{\tilde{D}}^{(b)}(x)
+\max(1-(\beta+\eta),0)\delta(x-1),
\label{eq:typwclemma2eq4}
\end{eqnarray}
with
\begin{eqnarray}
f_{\tilde{D}}^{(b)}(x)   &  =  &
\begin{cases}
	\frac{\sqrt{-(x-(\beta+\eta))^2-4\beta\eta(x-1)}}{2\pi(x-x^2)}, & \mbox{if }  x_l\leq x\leq x_u. \\
	0, & \mbox{otherwise}.
\end{cases}
\label{eq:typwclemma2eq5}
\end{eqnarray}
\end{lemma}
\begin{proof}
Follows through a combination of (\ref{eq:typwcanl44}), (\ref{eq:typwcanl46}),  (\ref{eq:typwcanl55}),  (\ref{eq:typwcanl68}), and the above discussion.
\end{proof}

In Figure \ref{fig:ftildeDbeta01eta08} we show the spectral function obtained based on the above lemma for $\beta=0.1$ and $\eta=0.8$. We observe a very strong agreement between the simulated results and the above theoretical predictions. Simulation results were obtained using moderately large $n=4000$.

\begin{figure}[h]
\centerline{\includegraphics[width=.7\columnwidth]{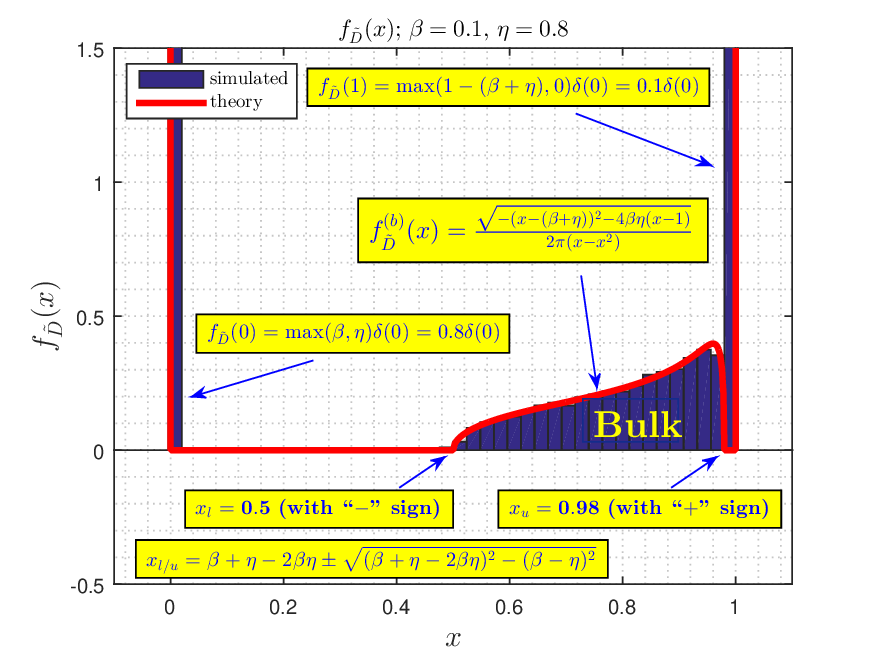}}
\caption{$f_{\tilde{D}}(x)$ -- spectral function of $\tilde{D}$; $\beta=0.1$ and $\eta=0.8$}
\label{fig:ftildeDbeta01eta08}
\end{figure}

In Figure \ref{fig:ftildeDbeta02eta09} we show the spectral function obtained based on the above lemma for $\beta=0.2$ and $\eta=0.9$. Due to a remarkable property of the underlying functions the spectrum is identical as in Figure \ref{fig:ftildeDbeta01eta08} apart from the fact that the multiplier of the delta function at zero is increased from $0.8$ to $0.9$ at the expense of removing the delta function at one. We also again observe a very strong agreement between the simulated results and the theoretical predictions. As in Figure \ref{fig:ftildeDbeta01eta08}, Simulation results were again obtained for $n=4000$.

\begin{figure}[h]
\centerline{\includegraphics[width=.7\columnwidth]{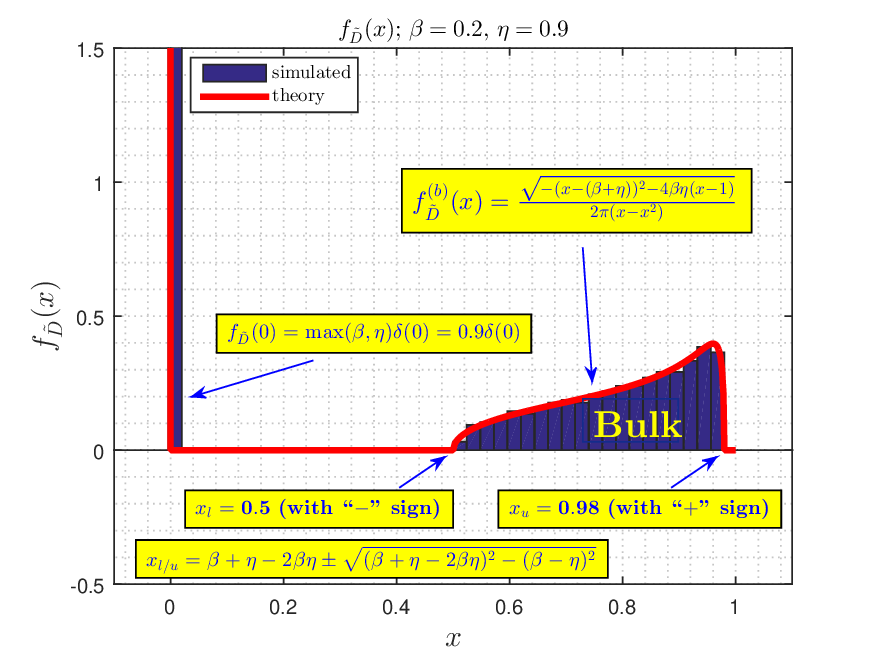}}
\caption{$f_{\tilde{D}}(x)$ -- spectral function of $\tilde{D}$; $\beta=0.2$ and $\eta=0.9$}
\label{fig:ftildeDbeta02eta09}
\end{figure}

Various other proofs of the lemma are possible. We, however, found the above one as very instructive. 

\section{Real data properties}

As we have mentioned in the introduction, it is very important to properly choose the model for the analyzed system. Since the algorithmic methodology that we analyzed applies to a wide range of causal inference settings (and especially so in the block causal inference), we want to ensure that the underlying model is properly selected so to represent a good fit to real data. The approximately low rank model that we consider is typically the preferred choice when the deviation from the ideal low rankness is caused by internal data properties rather by the errors in observations. A typical example where that happens is the so-called California tobacco study from \cite{ADHsynth10,AbaCat2018} (a golden standard, most often used in real data studies in recent years within the synthetic control methods).

Namely, as we discussed in the paper, in this example one postulates that it is likely that over years there is a similar trend in tobacco consumption across a set of different states in North America. This in turn implies that a matrix containing (in each row) annual consumption data of each state should be approximately low rank. In Figure
\ref{fig:CSEsingval}, we plot the spectrum of the resulting tobacco consumption matrix. As it can be seen from the figure, the spectrum clearly indicates the approximately low rank properties. Magnitudes of several largest singular values are substantially larger than magnitudes of the remaining ones (the matrix also happens to have a strong spiked component that doesn't change anything regarding the overall approximate low rankness; in fact it actually strengthens it).

For completeness, we do mention that the matrix of interest is of size $38\times 31$ and that it contains the annual average cigarette consumption for $38$ states over a span of $31$ years (this basically implies that the size of $X_{sol}$ is $38\times 31$ and the total number of singular values is $31$ as the figure indicates). We once again point out that one should observe that the approximate low rankness comes as a consequence of the internal structure of the data rather than as a consequence of incorrectly recorded data.

\begin{figure}[ht]
\vskip 0.2in
\begin{minipage}[b]{.7\linewidth}
\begin{center}
	\centerline{\includegraphics[width=1.0\columnwidth]{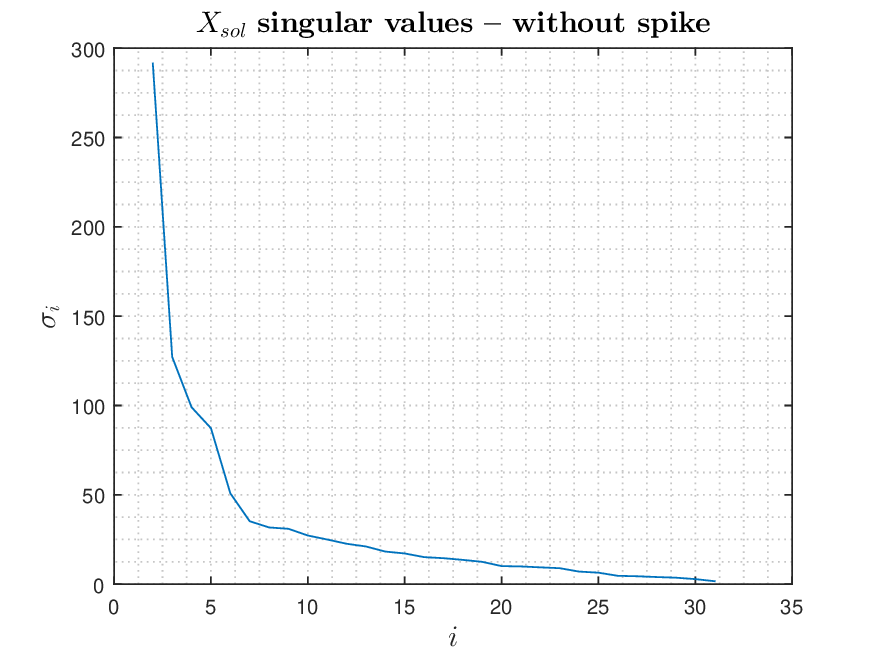}}
\end{center}
\end{minipage}
\begin{minipage}[b]{.7\linewidth}
\begin{center}
	\centerline{\includegraphics[width=1.0\columnwidth]{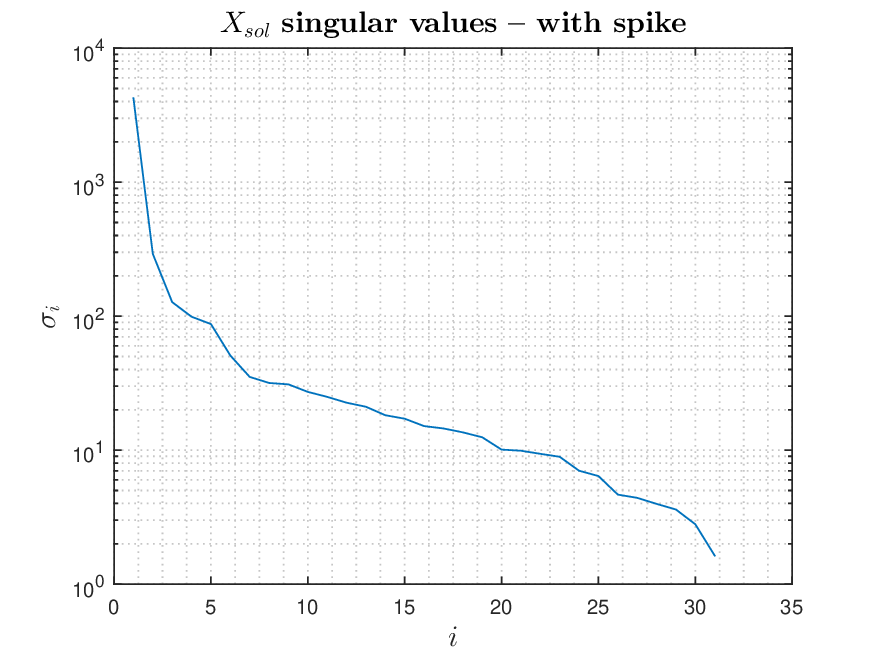}}
\end{center}
\end{minipage}
\vspace{-.35in}\caption{California smoking data set. We plot the spectrum of $X_{sol}$ (without and with the largest component) -- (upper part) \underline{without} rank one spike; (lower part) \underline{with} rank one spike. One of the singular values is substantially larger than any of the remaining ones (the largest singular value is above $4000$, whereas all others are smaller than $300$).}
\vskip -0.05in
\label{fig:CSEsingval}
\end{figure}



\bibliographystyle{imsart-number} 
\bibliography{cinfspecidealRefs1v2}       


\end{document}